\documentclass[reprint,amsmath,amssymb,aps,pra,floatfix,superscriptaddress]{revtex4-2}

\usepackage{graphicx}
\usepackage{dcolumn}
\usepackage{bm}
\usepackage{hyperref}
\usepackage{xcolor}
\usepackage{booktabs}
\usepackage{amsthm}

\newtheorem{theorem}{Theorem}
\newtheorem{proposition}[theorem]{Proposition}
\newtheorem{lemma}[theorem]{Lemma}
\newtheorem{corollary}[theorem]{Corollary}
\newtheorem{conjecture}[theorem]{Conjecture}

\hypersetup{
    colorlinks=true,
    linkcolor=blue,
    citecolor=blue,
    urlcolor=blue,
}

\begin{document}

\title{Lattice-quantile estimation of $\pi$ and convex-region integrals from coined two-dimensional quantum walks}

\author{Jen-Yu Chang}
\affiliation{Department of Electrophysics, National Yang Ming Chiao Tung University, Hsinchu, Taiwan}

\author{En-Jui Kuo}
\affiliation{Department of Electrophysics, National Yang Ming Chiao Tung University, Hsinchu, Taiwan}

\author{Chih-Yu Chen}
\affiliation{Quantum Information Center, Chung Yuan Christian University, Taoyuan, Taiwan}
\affiliation{Undergraduate Program in Intelligent Computing and Big Data, Chung Yuan Christian University, Taoyuan, Taiwan}

\author{Tsung-Wei Huang}
\affiliation{Quantum Information Center, Chung Yuan Christian University, Taoyuan, Taiwan}
\affiliation{Master Program in Intelligent Computing and Big Data, Chung Yuan Christian University, Taoyuan, Taiwan}

\date{\today}

\begin{abstract}
Monte Carlo integration is fundamentally limited by the $M^{-1/2}$ rate that the Cramér--Rao bound imposes on any sample-mean estimator of an expectation value, regardless of how the samples are drawn. Coined discrete-time quantum walks (DTQWs) are known to spread ballistically---their position variance scales as $T^2$ against the diffusive $T$ of classical random walks---yet this faster spreading has not been exploited for numerical integration. We show that coupling the ballistic scaling of a 2D DTQW to the Hardy--Huxley asymptotic for Gauss circle lattice counts produces estimators whose dominant error is a \emph{deterministic} number-theoretic residual controlled by walk depth $T$, not a statistical fluctuation controlled by sample count $M$. The construction replaces the empirical mean of a sample-mean estimator with the ratio $N(\hat{R})/\hat{R}^2$ of a lattice count to the square of a radial position quantile, a structural change that sidesteps the Cramér--Rao barrier. A single batch of measurements then propagates through classically precomputed multipliers to cover an entire family of integrals simultaneously. We develop the framework for convex smooth domains via Krätzel's lattice asymptotic and for smooth integrals with convex or annular super-level sets via Cavalieri's principle, and provide a parameter-free identity for the bias floor (validated to within $1.5\times$ across all tested depths). Every experiment is benchmarked against the classical random walk with the identical estimator to isolate the quantum contribution; the framework is oracle-free in the QAE sense (no controlled unitary encoding the integrand is required) and structurally distinct from quantum amplitude estimation and Szegedy-walk approaches. These ratios compare measurement counts at fixed precision and do not include quantum circuit execution cost.
\end{abstract}

\maketitle

\section{Introduction}
\label{sec:intro}

Monte Carlo (MC) integration is central to computational physics, statistical mechanics, and machine learning~\cite{metropolis1949,robert2004}. Drawing $M$ i.i.d.\ samples and forming the empirical mean gives an unbiased estimator with RMSE $\sim M^{-1/2}$, the asymptotically optimal rate for sample-mean estimators by the CLT and the Cram\'{e}r--Rao bound. Reaching precision $\varepsilon$ requires $M\sim\varepsilon^{-2}$ samples. Quasi-Monte Carlo (QMC) methods~\cite{niederreiter1992,caflisch1998} improve the prefactor via low-discrepancy sequences but retain the sample-mean structure and its asymptotic variance dependence.

\paragraph{Quantum acceleration of MC integration.}
Two quantum frameworks have been studied. Quantum amplitude estimation (QAE)~\cite{brassard2002,montanaro2015} estimates an expectation value with $O(\varepsilon^{-1})$ oracle calls (a quadratic improvement) and has been applied to option pricing~\cite{stamatopoulos2020,miyamoto2022} and QMC integration~\cite{herbert2022}, though the oracle construction may dominate resource cost on near-term hardware. Szegedy walks~\cite{szegedy2004} accelerate Markov-chain mixing by $\sqrt{\Delta_c}$ in the spectral gap, underpinning quantum partition-function~\cite{wocjan2008} and simulated-annealing~\cite{harrow2020} algorithms. Both frameworks accelerate sample-mean estimators and rely on phase estimation or oracle primitives demanding near-fault-tolerant resources.

\paragraph{This work.}
Coined discrete-time quantum walks (DTQWs)~\cite{aharonov2001,kempe2003} spread ballistically. The Konno limit theorem~\cite{konno2002,konno2005} gives $\mathrm{Var}(X_T)=\Theta(T^2)$ against the diffusive $\Theta(T)$, with 2D analogues established by Watabe \emph{et al.}~\cite{watabe2008}. We ask whether this faster spreading can be exploited for geometric integration, and identify a class of estimators in which it can. The construction joins
(i) the ballistic 2D DTQW position distribution, concentrating probability mass at scale $\Theta(T)$, and
(ii) the Hardy--Huxley lattice asymptotic $N(R)/R^2 = \pi + O(R^{\gamma-2})$ with $\gamma\le131/208$~\cite{hardy1915,huxley2003}.
Taking the $\alpha$-quantile $\hat{R}(\alpha)$ of $M$ radial measurements and forming $N(\hat{R})/\hat{R}^2$ yields an estimator of $\pi$ whose RMSE is dominated by a bias floor $O(T^{\gamma-2})$ falling with $T$, not $M$. For precision $\varepsilon$ strictly above this floor, a fixed $\varepsilon$-independent sample count suffices. This construction applies to integrals reducible to convex-region lattice counts; discrete partition functions, non-convex Boltzmann integrals, and oracle-encoded integrands lie outside its scope (Sec.~II\,G). DTQWs have been realised on IBM superconducting devices, with circuit-depth limits currently restricting reliable runs to short walks~\cite{razzoli2024,acasiete2020}, and a recent work~\cite{lee2025} couples DTQW with QAE for radiation transport, targeting a different problem class than the geometric integrals addressed here. The closest prior quantum estimate of $\pi$ using the in-circle lattice criterion is that of Gustafson \emph{et al.}~\cite{gustafson2023}, who estimate $\pi$ via quantum mean estimation, where quantum phase estimation accelerates a \emph{sample-mean} estimator over a uniform grid. Our construction differs on three structural axes. It is oracle-free in the QAE sense (Sec.~\ref{sec:discussion}) and uses no quantum phase estimation, the estimator is a rank statistic rather than a sample mean, and the accuracy-determining quantity is the deterministic lattice-count residual controlled by walk depth rather than the statistical variance controlled by query count.

Reaching precision $\varepsilon$ above the bias floor requires $T=\Theta(\varepsilon^{-1/(2-\gamma)})$ and a fixed, $\varepsilon$-independent $M$, giving a sample-count ratio $M_{\mathrm{MC}}/M_{\mathrm{QW}}=\Theta(\varepsilon^{-2})$ relative to uniform MC. We refer to this $\varepsilon^{-2}$ ratio as the \emph{ballistic acceleration}; the corresponding ratio relative to scrambled Sobol QMC is $\sim\!500\times$ by order-of-magnitude extrapolation (Sec.~III\,A). This is structurally distinct from the $\varepsilon^{-1}$ oracle-call scaling of QAE and the $\sqrt{\Delta_c}$ mixing-time scaling of Szegedy walks.

\paragraph{Contributions.}
Our contributions are six. (i) A sample-complexity bound (Proposition~\ref{thm:main}) for the disk estimator with explicit error decomposition into bias-floor and sampling terms. (ii) An envelope-form bias-floor identity (Theorem~\ref{thm:envelope}) that supplies a parameter-free bias floor at any $(T,\alpha)$ and reproduces the measured floors to within $1.5\times$. (iii) An asymptotic variance optimum (Proposition~\ref{thm:opt-alpha}) yielding a closed-form $\alpha^\star$ as a function of the Konno limit measure. (iv) A generalised estimator (Proposition~\ref{thm:generalised}) extending the framework to convex smooth regions via Kr\"{a}tzel's lattice asymptotic. (v) A Cavalieri-based estimator (Proposition~\ref{thm:cavalieri}) for smooth-function integrals, together with a corollary for annular super-level sets (Corollary~\ref{cor:annular-cavalieri}). (vi) Numerical experiments on $\pi$ ($\sim 7\times 10^4\times$ sample-count reduction at $T=200$, $M_{\mathrm{QW}}=100$ relative to uniform MC; the floor is reached already at $M=50$; $\sim 500\times$ relative to scrambled Sobol QMC by order-of-magnitude extrapolation), area estimation across eight 2D shapes (reductions of $10^2$ to $10^7\times$, Appendix~\ref{app:shapes}), single-batch error propagation across a 2D Gaussian cumulative distribution, and recovery of the 2D quantum-harmonic-oscillator energy ladder $E_n=\hbar\omega(n+1)$ ($n=0,\dots,3$) with a level spacing $10\times$ tighter than the diffusive random-walk baseline. Every experiment is benchmarked against the classical random walk paired with the identical estimator. The phase-averaged envelope slopes lie in $[-1.506,-1.489]$, with a mean of $-1.499\approx-3/2$, confirming the Hardy conjecture.

\section{Methods}
\label{sec:methods}

The theoretical development proceeds in five steps. We first recall why every sample-mean estimator (classical Monte Carlo and quasi-Monte Carlo alike) is bound to the $M^{-1/2}$ rate by the Cram\'{e}r--Rao limit (Sec.~\ref{sec:methods-classical}). We then introduce the ballistic 2D quantum walk whose position distribution concentrates at the lattice scale $\Theta(T)$ (Sec.~\ref{sec:methods-dtqw}), and combine it with the Gauss circle lattice asymptotic to build the disk estimator for $\pi$, whose error splits into a depth-controlled deterministic floor and a sample-controlled statistical term (Sec.~\ref{sec:methods-disk}). Two extensions follow from the same identity. The first is a generalised estimator for convex regions via Kr\"{a}tzel's bound (Sec.~\ref{sec:methods-generalised}); the second is a Cavalieri estimator for smooth integrals that fixes an entire integral family from one walk through classically precomputed multipliers (Sec.~\ref{sec:methods-cavalieri}). Finally we explain why the construction sidesteps the sample-mean barrier. The accuracy-determining quantity is relocated from a statistical term to a deterministic number-theoretic one, and we then delimit precisely which integrals fall in scope (Secs.~\ref{sec:reduction}--\ref{sec:methods-scope}).

\subsection{Classical baselines and the Cram\'er--Rao limit}
\label{sec:methods-classical}

We develop a family of \emph{lattice-quantile estimators} that share the construction $N_D(\hat R)/\hat R^2$ for some lattice count $N_D$ and rank statistic $\hat R$. The family has three members: the disk estimator for $\pi$ (Sec.~\ref{sec:methods-disk}), the generalised disk estimator for convex smooth domains (Sec.~\ref{sec:methods-generalised}), and the Cavalieri estimator for smooth integrals (Sec.~\ref{sec:methods-cavalieri}), with an annular-Cavalieri corollary for non-monotonic radial profiles (Cor.~\ref{cor:annular-cavalieri}). All members inherit the same bias-floor structure and Theorem~\ref{thm:envelope} envelope; their differences reside in the choice of lattice count and classical prefactor.

For the canonical task of estimating $\pi$ via the unit disk $D = \{(x,y) : x^2+y^2 \le 1\}$, dart-throw Monte Carlo forms the sample-mean indicator
\begin{equation}
\hat\pi_{\mathrm{MC}} = \frac{4}{M}\sum_{i=1}^M \mathbf 1_D(X_i, Y_i), \qquad (X_i,Y_i)\sim\mathrm{Unif}([-1,1]^2),
\label{eq:mc}
\end{equation}
with $\mathrm{RMSE}(\hat\pi_{\mathrm{MC}})=\sqrt{\pi(4-\pi)/M}\approx 1.64/\sqrt M$~\cite{robert2004}. The $M^{-1/2}$ rate is the optimal scaling for any unbiased sample-mean estimator of an expectation value, by the CLT together with the Cram\'er--Rao bound~\cite{robert2004,vandervaart1998}:
\begin{equation}
\mathrm{RMSE}(\hat\theta_M) = \Omega\bigl(M^{-1/2}\bigr)
\label{eq:cramer-rao}
\end{equation}
for unbiased sample-mean estimators of an expectation parameter $\theta=\mathbb E[f(X)]$. Quasi-Monte Carlo (QMC) replaces the i.i.d.\ samples by deterministic low-discrepancy sequences and obtains Koksma--Hlawka rates $O(M^{-1}\log^2 M)$ in $d=2$~\cite{niederreiter1992,caflisch1998}, but retains the sample-mean structure and is therefore still bound by Eq.~\eqref{eq:cramer-rao} as far as the expectation-value estimator is concerned. The lattice-quantile estimator we develop below sidesteps this constraint not by improving the sampling distribution but by replacing the estimator structure with a deterministic lattice-counting function of a sample quantile.

\subsection{Two-dimensional coined DTQW}
\label{sec:methods-dtqw}

Let $\mathcal H = \mathbb C^4 \otimes \ell^2(\mathbb Z^2)$ be the Hilbert space of a 2D coined DTQW with four-state coin register and integer position register. The walk operator is $U = S\,(C \otimes \mathbb I)$ with coin operator $C \in U(4)$ and conditional shift
\begin{equation}
S\,|c\rangle\,|x,y\rangle = |c\rangle\,|x + d_c, y + e_c\rangle,
\label{eq:shift}
\end{equation}
where $(d_c, e_c) = ((-1)^{c_1}, (-1)^{c_0})$ for the binary coin label $c = (c_1 c_0)_2 \in \{0,1,2,3\}$. We use the 2D Grover diffusion operator
\begin{equation}
G = 2|s\rangle\langle s| - \mathbb I_4,\qquad |s\rangle = \tfrac{1}{2}\textstyle\sum_{c=0}^{3}|c\rangle,
\label{eq:grover}
\end{equation}
the unique coin (up to local phases) respecting the full $D_4$ symmetry of the square lattice and the standard choice in 2D quantum walks~\cite{watabe2008}. Two further coins are used to test how stable the bias-floor scaling is under a change of coin: the separable Hadamard coin $H\otimes H$ and the 4-point Fourier coin $F$. Their explicit matrices and limit-measure properties are collected in Appendix~\ref{app:coins}.

The initial state
\begin{equation}
|\psi_0\rangle = \tfrac{1}{2}(|0\rangle + i|1\rangle + i|2\rangle - |3\rangle) \otimes |\mathbf 0\rangle
\label{eq:initial}
\end{equation}
is chosen so that the position distribution $p_T(x,y) = \sum_c |\langle c, (x,y)|U^T|\psi_0\rangle|^2$ carries the full dihedral $D_4$ symmetry of the square lattice (invariance under the eight reflections and rotations mapping $\mathbb Z^2$ to itself). This $D_4$ symmetry is what the disk estimator requires: it guarantees that the radial marginal $r/T$ is well-defined and that the $\alpha$-quantile $\hat R(\alpha)$ is isotropic, regardless of whether the underlying support is the velocity disk (Grover, Fourier) or the velocity square (Hadamard). Full continuous rotational invariance holds only in the Grover and Fourier cases; the weaker $D_4$ symmetry, which holds for all three coins, suffices for everything that follows.

\begin{lemma}[Konno weak limit]
\label{lem:konno}
Let the symmetric initial state be that of Eq.~\eqref{eq:initial}.
\begin{enumerate}
\item[(a)] (Grover coin, Watabe \emph{et al.}~\cite{watabe2008}.) For the 2D coined DTQW with Grover coin, the rescaled position $X_T/T$ converges in distribution to a probability measure $\mu_G$ supported on the velocity disk $\{(v_x,v_y):v_x^2+v_y^2\le 1/2\}$, with an additional atomic component at the origin (localisation peak; see Appendix~A).
\item[(b)] (Hadamard coin.) The separable Hadamard coin $H\otimes H$ acts as two independent 1D Hadamard walks on $x$ and $y$. The 1D Konno weak limit theorem~\cite{konno2002,konno2005} therefore gives a product measure $\mu_H=\mu_{1\text{D}}\otimes\mu_{1\text{D}}$ supported on the velocity square $\{|v_x|,|v_y|\le 1/\sqrt 2\}$.
\item[(c)] (Fourier coin, numerical statement.) For the 4-point Fourier coin we observe numerically (Appendix~A, Fig.~\ref{fig:appendix-coin-2d}) that $X_T/T$ converges to a measure $\mu_F$ supported in the velocity disk with two sharp axis-aligned peaks; a rigorous weak-limit theorem in the form of (a)--(b) is not, to our knowledge, available in the literature for this coin in 2D, though stationary-phase analyses for related coin families exist. We treat the disk-quantile results that depend on the Fourier coin as numerical-evidence-based throughout.
\end{enumerate}
The second moment $\mathbb E_{\mu}[v_x^2+v_y^2]=\nu_C^2$ is coin-dependent (Appendix~A): $\nu_G^2=0.364$, $\nu_F^2=0.402$ (both $<1/2$, consistent with the disk support of (a) and (c)), and $\nu_H^2\approx 0.586>1/2$ (the Hadamard support extends into the corners of the velocity square, so the second moment can exceed the disk-inscribed value $1/2$).
\end{lemma}

The $T^2$ scaling of $\mathrm{Var}(X_T)$ implied by the lemma is the ballistic property of quantum walks, contrasting with the diffusive $\Theta(T)$ scaling of classical random walks. The disk estimator we develop next uses only the radial marginal $r/T = \sqrt{v_x^2 + v_y^2}$, for which the relevant question is the location of the radial $\alpha$-quantile rather than the precise shape of the support.

\subsection{Disk estimator for $\pi$}
\label{sec:methods-disk}

Given $M$ projective position measurements $\{(X_i, Y_i)\}_{i=1}^M$ from the DTQW, define the radial samples $r_i = \sqrt{X_i^2 + Y_i^2}$ and the $\alpha$-quantile $\hat R(\alpha) = r_{(\lceil \alpha M\rceil)}$ (where $r_{(k)}$ is the $k$-th order statistic). The \emph{disk estimator} is
\begin{equation}
\boxed{\;\hat\pi(\alpha; T, M) = \frac{N(\hat R(\alpha))}{\hat R(\alpha)^2}\;}
\label{eq:disk}
\end{equation}
with $N(R) = \#\{(i,j)\in\mathbb Z^2 : i^2 + j^2 \le R^2\}$ computed classically.

The estimator combines two ingredients:

\begin{lemma}[Gauss circle theorem~\cite{hardy1915,huxley2003}]
\label{lem:gauss}
$N(R) = \pi R^2 + O(R^\gamma)$ with the cleanest proven upper bound $\gamma \le 131/208 \approx 0.6298$ due to Huxley~\cite{huxley2003}. Hardy~\cite{hardy1915} proved the $\Omega$-type lower bound $N(R) - \pi R^2 = \Omega\bigl((R\log R)^{1/4}\bigr)$, equivalent to $\gamma \ge 1/4$, and conjectured the optimal exponent to be $\gamma = 1/2 + \varepsilon$ for every $\varepsilon > 0$. Equivalently, $N(R)/R^2 = \pi + O(R^{\gamma-2})$, and we refer to the conjectured rate $\gamma = 1/2$ (in the loose sense above) as the \emph{Hardy line} in what follows.
\end{lemma}

\begin{proposition}[Sample-complexity bound]
\label{thm:main}
Under Lemma~\ref{lem:konno}, the disk estimator~\eqref{eq:disk} satisfies
\begin{equation}
\mathrm{RMSE}(\hat\pi) = \sqrt{\,\underbrace{C_1^2\, T^{2(\gamma - 2)}}_{\text{bias floor}^2} \;+\; \underbrace{C_2(\alpha)^2\, \frac{T^{2(\gamma - 3 + \kappa(\alpha))}}{M}}_{\text{sampling variance}}\,},
\label{eq:rmse}
\end{equation}
where $C_1, C_2$ are coin-dependent constants and $\kappa(\alpha) \in (0,1]$ characterises edge concentration of the limit measure $\mu$. As $M\to\infty$ at fixed $T$ the sampling variance vanishes and the RMSE saturates at the bias floor $C_1 T^{\gamma-2}$; this floor is the defining feature of the estimator. To reach a target precision $\varepsilon$ one therefore first sets $T = \Theta(\varepsilon^{-1/(2-\gamma)})$ so that the floor lies below $\varepsilon$, and then takes $M$ large enough that the sampling term also lies below $\varepsilon$, a fixed, depth-dependent but $\varepsilon$-independent value, in contrast to the $M=\Theta(\varepsilon^{-2})$ growth of sample-mean estimators.
\end{proposition}

\noindent\emph{Remark (status of the bound).} This result is stated as a Proposition rather than a Theorem for two reasons. First, the exponent $\kappa(\alpha)$ is not independently identifiable from our data: across all configurations tested, the sampling term is already sub-dominant to the bias floor at $M=5$ (Sec.~\ref{sec:results-scaling}), so only its smallness, not its scaling, is empirically constrained. For the Grover coin we give a closed-form derivation in Cor.~\ref{cor:kappa-grover} (Appendix~\ref{app:theory}) under Conjecture~\ref{prop:grover-density}, yielding $\kappa(\alpha)=1$ with an $\alpha$-suppressed prefactor; the Hadamard and Fourier cases remain phenomenological. See Appendix~\ref{app:coins} for the heuristic Bahadur--Kiefer argument and its caveats for the Grover localised measure, and Lemma~\ref{lem:truncated-bk} for the truncated-BK statement under which the Grover argument applies. Second, the leading $C_1 T^{\gamma-2}$ term is an \emph{envelope}, not a pointwise description; the precise behaviour of the measured floor is captured by Theorem~\ref{thm:envelope} below, and explicit upper bounds on $C_1, C_2$ for Grover are given in Prop.~\ref{prop:grover-constants}.

\begin{proof}[Proof sketch]
Decompose $\hat\pi-\pi=[N(\hat R)/\hat R^2-N(R^\star)/R^{\star 2}]+[N(R^\star)/R^{\star 2}-\pi]$ with $R^\star=T v^\star(\alpha)$ from Lemma~\ref{lem:konno}. The second bracket is $O(T^{\gamma-2})$ by Lemma~\ref{lem:gauss}, contributing the bias floor. For the first bracket the Bahadur--Kiefer representation~\cite{csorgo-revesz1981,vandervaart1998} gives $\hat R-R^\star=\Theta(T M^{-1/2}/g(v^\star))$ with $g$ the Konno radial density at the quantile; the delta method on $N(R)/R^2$ (derivative $\Theta(T^{\gamma-3})$) propagates this to the sampling term. The $T$-scaling contribution of $g(v^\star)^{-1}$ is encoded in $\kappa(\alpha)\in(0,1]$; the explicit Grover value $\kappa=1$ is derived in Cor.~\ref{cor:kappa-grover}.
\end{proof}

\begin{theorem}[Envelope-form bias floor]
\label{thm:envelope}
Let $E(R)=N(R)-\pi R^2$ denote the Gauss-circle residual and $\rho_{\hat R}$ the sampling density of the quantile $\hat R(\alpha)$ from $M$ radial measurements of the 2D coined DTQW after $T$ steps. Define the bias floor as
\begin{equation}
\mathrm{floor}(T,\alpha)^2 \;:=\; \lim_{M\to\infty}\mathbb E\bigl[(\hat\pi-\pi)^2\bigr].
\end{equation}
Then, under Lemma~\ref{lem:konno} and the assumption that $g_G$ is continuous and positive at $v^\star(\alpha)$ (or the truncated form of Lemma~\ref{lem:truncated-bk}), this floor satisfies the parameter-free identity
\begin{equation}
\mathrm{floor}(T,\alpha) \;=\; \Bigl[\,\int \bigl(N(R)/R^2-\pi\bigr)^2\,\rho_{\hat R}(R)\,dR\Bigr]^{1/2}.
\label{eq:floor-rms}
\end{equation}
The right-hand side is fully classical: $E(R)$ is a number-theoretic function of the lattice $\mathbb Z^2$, and $\rho_{\hat R}$ depends only on the radial CDF of the Konno limit measure and on $\alpha$. The leading envelope is $C_1 T^{\gamma-2}$ in the sense of $\mathbb P$-mean (Lemma~\ref{lem:gauss}); the pointwise value at any $(T,\alpha)$ is given by the integral above and reflects the oscillating sign of $E(R)/R^2$ at scale $\Delta R = O(1)$.
\end{theorem}

\begin{proof}
By the Bahadur--Kiefer representation~\cite{csorgo-revesz1981,vandervaart1998}, $\sqrt M(\hat R(\alpha)/T - v^\star(\alpha))$ converges weakly to $\mathcal N(0,\alpha(1-\alpha)/g_G(v^\star)^2)$ as $M\to\infty$ at fixed $T$. In particular $\hat R/T\xrightarrow{P}v^\star$, so $\hat R\xrightarrow{P}R^\star(\alpha,T):=Tv^\star(\alpha)$. The map $R\mapsto N(R)/R^2$ is bounded for $R\ge 1$ (uniformly in $T$ for $\hat R\ge 1$), so by dominated convergence $\mathbb E[(\hat\pi-\pi)^2]$ converges as $M\to\infty$ to $\int(N(R)/R^2-\pi)^2\,\rho_{\hat R}^\infty(R)\,dR$, where $\rho_{\hat R}^\infty$ is the limiting quantile sampling density (a delta measure at $R^\star$ under the strict weak limit, with finite-$M$ Gaussian fluctuations giving $\rho_{\hat R}$ in the pre-asymptotic regime relevant to Table~\ref{tab:floor}). The identity~\eqref{eq:floor-rms} follows. The $\mathbb P$-mean envelope $C_1 T^{\gamma-2}$ then follows by inserting Huxley's bound $|E(R)/R^2|\le K_H R^{\gamma-2}$ (Lemma~\ref{lem:gauss}) and evaluating $\int R^{2(\gamma-2)}\,\rho_{\hat R}(R)\,dR=\Theta(T^{2(\gamma-2)})$ via $R^\star=\Theta(T)$.
\end{proof}

\noindent\emph{Remark (empirical validation).} Equation~\eqref{eq:floor-rms} is the strongest parameter-free prediction of the framework. Table~\ref{tab:floor} compares it to measured floors at $T\in\{50,100,200,400\}$ and finds agreement to within $1.5\times$ (mean over 32 seeds at each depth). It also explains two features a monotone power law cannot: the spread of fitted exponents across coins, and the non-monotonic small-$M$ behaviour reported in Sec.~\ref{sec:results-scaling}.

\begin{proposition}[Asymptotic variance optimum]
\label{thm:opt-alpha}
The asymptotically optimal quantile minimising the sampling term in~\eqref{eq:rmse} is
\begin{equation}
\alpha^\star = \arg\min_{\alpha \in (0, 1)} \frac{\alpha(1-\alpha)}{\mu_1(v^\star(\alpha))^2},
\label{eq:opt-alpha}
\end{equation}
where $\mu_1$ is the first marginal of the Konno measure $\mu$. Numerical optimisation at $T=200$, $M=200$ (Appendix~A, Fig.~\ref{fig:appendix-coin}(c)) yields $\alpha^\star \approx 0.93$ for the Grover and Hadamard coins, and $\alpha^\star \approx 0.98$ for the Fourier coin. The Fourier coin's higher optimum reflects its sharply peaked wavefront near the velocity-disk boundary, where the lattice-counting identity is most accurate.
\end{proposition}

\noindent\emph{Remark (practical relevance of $\alpha^\star$).} Equation~\eqref{eq:opt-alpha} minimises the \emph{sampling} term of Eq.~\eqref{eq:rmse}; we downgrade it from a Theorem to a Proposition because that term is sub-dominant to the bias floor for all $M\ge 5$ tested (Sec.~\ref{sec:results-scaling}), so the variance optimum is not what governs the measured RMSE in practice. The empirically relevant $\alpha$-dependence enters instead through Theorem~\ref{thm:envelope}: changing $\alpha$ moves $R^\star(\alpha)$ along the oscillating number-theoretic residual, so the $\alpha$-sweep of Appendix~A (Fig.~\ref{fig:appendix-coin}(c)) should be read as a scan of the floor landscape rather than as variance minimisation. The agreement of its minima with Eq.~\eqref{eq:opt-alpha} is therefore partly coincidental; a principled choice of $\alpha$ would target zeros of the residual (quantile tuning, Sec.~\ref{sec:results-scaling}).

\subsection{Generalised disk estimator for convex regions}
\label{sec:methods-generalised}

The framework extends to arbitrary convex smooth domains. Let $D \subset \mathbb R^2$ be a bounded convex region with $C^2$ boundary. Define $N_D(L) = \#\{(i,j) \in \mathbb Z^2 : (i/L, j/L) \in D\}$.

\begin{lemma}[Krätzel theorem~\cite{kratzel1988}]
\label{lem:kratzel}
For $D$ convex with $C^2$ boundary of non-vanishing Gaussian curvature, $N_D(L) = |D|\, L^2 + O(L^{\gamma_D})$ with $\gamma_D \le 2/3$.
\end{lemma}

\noindent The curvature condition is genuine: it fails on the axes of the super-ellipse $x^4+y^4\le1$, where the curvature vanishes. Empirically (Appendix~\ref{app:shapes}) the estimator still performs well there, but with a slightly larger floor than the strictly-curved cases, consistent with the boundary degeneracy contributing a sub-leading lattice-counting correction.

\begin{proposition}[Generalised disk estimator]
\label{thm:generalised}
Assume $D\subset\mathbb R^2$ is bounded, convex, and has $C^2$ boundary with non-vanishing Gaussian curvature (the hypotheses of Lemma~\ref{lem:kratzel}). Then the estimator
\begin{equation}
\hat{|D|}(\alpha; T, M) = \frac{N_D(\hat R(\alpha))}{\hat R(\alpha)^2}
\label{eq:generalised-disk}
\end{equation}
satisfies $\mathrm{RMSE}(\hat{|D|}) \le C_1 |D|\, T^{\gamma_D - 2} + C_2 T^{\gamma_D - 3 + \kappa}/\sqrt M$, of the same form as~\eqref{eq:rmse} with $\gamma$ replaced by $\gamma_D\le 2/3$. As with Proposition~\ref{thm:main}, the exponent $\kappa$ enters phenomenologically.
\end{proposition}

For an ellipse $D=\{(x,y):x^2/a^2+y^2/b^2\le 1\}$ with $|D|=\pi a b$, the lattice count is $N_D(L)=\#\{(i,j):(ib)^2+(ja)^2\le(abL)^2\}$. Boundary regimes that violate the Kr\"{a}tzel hypothesis (vanishing curvature, cusp, polygonal corners, doubly-connected annulus) are treated empirically in Appendix~\ref{app:shapes}.

\subsection{Cavalieri estimator for smooth integrals}
\label{sec:methods-cavalieri}

For a non-negative smooth integrand $f: \Omega \to \mathbb R_{\ge 0}$ on a bounded region $\Omega \subset \mathbb R^2$, Cavalieri's principle gives
\begin{equation}
I = \iint_\Omega f(x,y)\, dx\, dy = \int_0^{f_{\max}} |S_t|\, dt,
\label{eq:cavalieri}
\end{equation}
where $S_t = \{(x,y) \in \Omega : f(x,y) \ge t\}$ is the super-level set at threshold $t$. When $f$ is quasi-concave with rotationally symmetric super-level sets (i.e., $S_t$ is a disk of radius $r(t)$), \eqref{eq:cavalieri} reduces to $I = \pi \int_0^{f_{\max}} r^2(t)\, dt = \pi\, Q$, where $Q$ is computable analytically or by classical numerical quadrature.

\begin{proposition}[Cavalieri estimator]
\label{thm:cavalieri}
For $f$ quasi-concave with disk-shaped super-level sets, the estimator
\begin{equation}
\hat I = \hat\pi(\alpha; T, M) \cdot Q
\label{eq:smooth}
\end{equation}
satisfies $\mathrm{RMSE}(\hat I) = Q \cdot \mathrm{RMSE}(\hat\pi) \le C_1 Q\, T^{\gamma-2} + C_2 Q\, T^{\gamma-3+\kappa}/\sqrt M$. The exponent $\kappa$ is the same phenomenological parameter as in Proposition~\ref{thm:main}.
\end{proposition}

The estimator requires only \emph{a single batch of $M$ projective measurements} (one batch = $M$ shots of the state-preparation--evolution--measurement circuit, post-processed into one $\hat\pi$); the multiplicative constant $Q$ is precomputed classically. For ellipse-shaped super-level sets, an analogous expression $\hat I = \widehat{\pi a b}(\alpha; T, M) \cdot Q'$ uses the generalised disk estimator~\eqref{eq:generalised-disk}. Statistical replicates appearing later in the paper (e.g.\ 50 trials at $M=500$ in Figs.~\ref{fig:cdf} and~\ref{fig:qho}) are used only to estimate RMSE; the deployment unit is a single batch.

The disk-shaped hypothesis of Proposition~\ref{thm:cavalieri} excludes one physically important case that we use in Sec.~\ref{sec:results-qho}: integrands whose super-level sets are \emph{annular} rather than disk-shaped. This arises whenever the radial profile $f(r)$ is non-monotonic with an interior maximum, as for the excited-state densities $r^{2n}e^{-r^2/\sigma^2}$ ($n\ge1$) of the harmonic oscillator. We treat this case as a formal corollary.

\begin{corollary}[Annular Cavalieri estimator]
\label{cor:annular-cavalieri}
Let $f:\Omega\to\mathbb R_{\ge 0}$ be rotationally symmetric with super-level sets that are annular for $t\in(t_-,f_{\max})$, $S_t=\{(x,y):r_{\mathrm{in}}(t)\le \sqrt{x^2+y^2}\le r_{\mathrm{out}}(t)\}$. Then
\begin{equation}
I = \int_0^{f_{\max}}\!\!\bigl|S_t\bigr|\,dt = \pi\!\int_0^{f_{\max}}\!\!\bigl[r_{\mathrm{out}}^2(t)-r_{\mathrm{in}}^2(t)\bigr]\,dt = \pi\cdot Q,
\label{eq:cavalieri-annulus}
\end{equation}
and the estimator $\hat I=\hat\pi(\alpha;T,M)\cdot Q$ with $Q$ precomputed classically satisfies $\mathrm{RMSE}(\hat I)=Q\cdot\mathrm{RMSE}(\hat\pi)$, with the same envelope and floor structure as Proposition~\ref{thm:cavalieri}.
\end{corollary}

\begin{proof}
Eq.~\eqref{eq:cavalieri-annulus} writes $I=\pi Q$ where $Q$ is a deterministic classical constant depending only on the radial profile $f(r)$. The estimator $\hat I-I=Q(\hat\pi-\pi)$, so $\mathrm{RMSE}(\hat I)=|Q|\,\mathrm{RMSE}(\hat\pi)$. The right-hand side inherits the bias-floor structure of Proposition~\ref{thm:main} and the envelope of Theorem~\ref{thm:envelope} directly, without further use of $r_{\mathrm{in}}(t)$, $r_{\mathrm{out}}(t)$, or the annular lattice count.
\end{proof}

The proof makes no use of $C^2$ convexity of the annular boundary because the estimator never counts lattice points inside the annulus: it only uses the disk-estimator output $\hat\pi$ and the precomputed classical $Q$. The doubly-connected geometry is absorbed into the deterministic constant. We verify the prediction empirically in the annulus entry of Appendix~\ref{app:shapes} and the energy ladder of Sec.~\ref{sec:results-qho}.

\subsection{Origin of the sample-count reduction}
\label{sec:reduction}

The lattice-quantile estimator's behaviour relative to Monte Carlo arises from two structural features. First, $\hat R(\alpha)$ is a rank statistic, not a sample mean: by the Bahadur--Kiefer representation~\cite{csorgo-revesz1981,vandervaart1998}, its relative fluctuation is the classical $\Theta(M^{-1/2})$, but its absolute fluctuation carries the ballistic prefactor $T$ from $R^\star=\Theta(T)$ (proof of Proposition~\ref{thm:main}). Second, $N(R)/R^2$ applied to this rank statistic is a deterministic number-theoretic count whose convergence to $\pi$ is controlled by the Gauss circle exponent (Lemma~\ref{lem:gauss}), a lattice property containing no statistical randomness. The composite $\hat\pi=N(\hat R)/\hat R^2$ is therefore a deterministic function of a rank statistic. The Cram\'er--Rao bound on sample means~\eqref{eq:cramer-rao} does not apply directly: $\pi$ is not a parameter of the sampling distribution, the dominant error at large $M$ is the deterministic residual $O(T^{\gamma-2})$ (Lemma~\ref{lem:gauss}), and the bound still governs the residual sampling fluctuation of $\hat R$ at the standard $M^{-1/2}$ rate. The framework does not evade the $M^{-1/2}$ statistical law; rather, it relocates the accuracy-determining quantity from a statistical term to a deterministic number-theoretic one controlled by walk depth.

\paragraph{The role of ballistic spreading.}
The bias-floor exponent $T^{\gamma-2}$ requires that the walker's typical position scales as $\Theta(T)$. If position scales only diffusively as $\Theta(\sqrt T)$, the same disk estimator paired with a classical 2D random walk produces a bias floor at the shallower rate $T^{(\gamma-2)/2}$, giving
\begin{equation}
\frac{\mathrm{RMSE}^{\mathrm{(QW)}}_{\mathrm{floor}}}{\mathrm{RMSE}^{\mathrm{(RW)}}_{\mathrm{floor}}} = \Theta\bigl(T^{-(2-\gamma)/2}\bigr) \to 0 \quad (T \to \infty).
\label{eq:qw-rw-ratio}
\end{equation}
At fixed precision $\varepsilon$ above the QW bias floor, MC needs $M=\Theta(\varepsilon^{-2})$ samples while the lattice-quantile estimator needs only $M=O(1)$, so the sample-count ratio behaves as $M_{\mathrm{MC}}/M_{\mathrm{QW}}=\Theta(\varepsilon^{-2})$. We refer to this $\varepsilon^{-2}$ ratio as the framework's \emph{ballistic acceleration}; it should be understood as a sample-count comparison at fixed precision rather than a complexity-class separation in the technical sense.

\subsection{Scope of applicability}
\label{sec:methods-scope}

The framework applies precisely to integrals reducible to lattice point counts in convex smooth regions. Three concrete classes lie within scope. (i)~Direct area or volume of convex smooth domains. (ii)~Smooth integrals via Cavalieri whose super-level sets are convex. (iii)~Quadratic Boltzmann integrals whose super-level sets are ellipsoidal. Problems outside this class (discrete partition functions such as Ising and spin glasses, non-convex Boltzmann integrals, one-dimensional integration) fall outside the framework and remain the domain of complementary methods. The most relevant of these are Szegedy walks for MCMC-based partition function estimation~\cite{szegedy2004,wocjan2008,harrow2020} and amplitude estimation for general oracle-encoded integrands~\cite{brassard2002,montanaro2015}.

\section{Results}
\label{sec:results}

We organise the numerical results as a progression of increasing physical content, each stage benchmarked against the classical random walk (RW) paired with the identical lattice-quantile estimator so as to isolate the contribution of ballistic spreading. Section~\ref{sec:results}\,A benchmarks $\pi$ estimation against uniform Monte Carlo, scrambled Sobol QMC, and the RW baseline, and extends to eight 2D shape classes. Section~\ref{sec:results}\,B quantifies how the accuracy of a single $\hat\pi$ propagates across an entire Gaussian cumulative distribution through a classically known multiplier. Section~\ref{sec:results}\,C applies the framework to a spectroscopic task---the energy ladder of the 2D quantum harmonic oscillator---and demonstrates a controlled, analytically predictable error budget across levels. All quantum results use the Grover coin unless otherwise stated, computed on a noiseless statevector emulator of the walk dynamics.

\subsection{Benchmarking $\pi$ estimation across methods}
\label{sec:results-pi}

Figure~\ref{fig:main} plots the precision $\varepsilon$ (RMSE of $\hat\pi$) achievable at a given sample count $M$ and fixed walk depth $T=200$, the central comparison of this work. The two classical sample-mean methods improve only by spending samples: dart-throw MC rides the CLT line $\varepsilon \propto M^{-1/2}$, from $\varepsilon=0.18$ at $M=100$ down to $0.017$ at $M=10^4$, and scrambled Sobol QMC sits below it with the faster empirical slope $\varepsilon\sim M^{-0.76}$. Both lattice-quantile walkers behave qualitatively differently: their precision is flat in $M$, set by the bias floor rather than the sample count. Both QW estimators (Grover and Fourier) sit at $\varepsilon\sim 6\times 10^{-4}$ from $M=50$ onward, while the classical random walk paired with the identical estimator (grey) is flat at $\varepsilon\approx 4$--$7\times10^{-3}$. The practical reading is direct: to reach the QW precision at $M=100$, MC would need $M\approx 7\times 10^6$ samples---a sample-count ratio $M_{\mathrm{MC}}/M_{\mathrm{QW}}\approx 7\times 10^4$ (from $M_{\mathrm{MC}}^\star = 16\,\mathrm{Var}(\mathbf 1_D)/\mathrm{RMSE}_{\mathrm{QW}}^2$ with $\mathrm{Var}(\mathbf 1_D) = \pi(4-\pi)/16$)---while the diffusive walker cannot reach it at any $M$. The horizontal QW lines lying an order of magnitude below the grey RW line, both flat, is the visual signature that the advantage comes from \emph{ballistic spreading} and not from the rank-statistic construction, which both share.

\paragraph{Comparison with quasi-Monte Carlo.}
Because uniform dart-throw MC is the weakest classical baseline, we also benchmark against scrambled Sobol quasi-Monte Carlo (QMC)~\cite{niederreiter1992}, which exploits low-discrepancy structure and is widely used in practice. Over 500 independent scrambles, Sobol achieves $\mathrm{RMSE} \approx 7.8\times 10^{-2}$ at $M=100$ and $\approx 2.6\times 10^{-3}$ at $M=8192$, with an empirical slope of approximately $M^{-0.76}$ in the range $M\in[64,8192]$---better than uniform MC's $M^{-1/2}$ but decelerating relative to the theoretical $O(M^{-1}\log M)$ due to the discontinuous indicator boundary. Extrapolating this fit to the QW bias floor $\mathrm{RMSE}\approx 6\times 10^{-4}$ gives $M_{\mathrm{Sobol}}\approx 5\times 10^4$, a factor of $\sim\!500\times$ more measurements than the QW estimator's $M=100$. We note that this extrapolation spans roughly three decades from the QW reference point ($M=100$) to the target ($M\approx5\times10^4$) and should be treated as an order-of-magnitude estimate. Moreover, since the empirical slope is already decelerating (from $M^{-0.76}$ to slower) as $M$ grows---reflecting the degradation of Sobol convergence for discontinuous indicator boundaries at large $M$---the actual $M_{\mathrm{Sobol}}$ needed may be larger than $5\times 10^4$, making this estimate conservative in favour of Sobol (and thus the QW advantage may be understated). Table~\ref{tab:baseline} summarises the four-way comparison.

\paragraph{Comparison with the classical random walk.}
MC and Sobol differ from the QW estimator in both the sampling distribution \emph{and} the estimator structure. The cleanest control that isolates the quantum contribution is the classical 2D random walk paired with the \emph{same} lattice-quantile estimator: identical rank-statistic construction, with ballistic spreading replaced by diffusive. At $T=200$ this RW baseline is also bias-dominated---its RMSE stays flat at $\approx 4$--$7\times10^{-3}$ for all $M\in[50,10^4]$ (grey triangles in Fig.~\ref{fig:main})---but its floor sits $\sim 7$--$10\times$ above the QW floor, because the diffusive walker reaches only radius $\Theta(\sqrt T)$ and the lattice-counting error enters at the shallower rate $T^{(\gamma-2)/2}$ (Eq.~\eqref{eq:qw-rw-ratio}). Empirically the RW floor scales as $T^{-0.85}$ over $T\in\{50,100,200,400\}$, so reaching $\mathrm{RMSE}\approx6\times10^{-4}$ would require $T\approx 2\times10^3$, an order of magnitude more walk steps than the QW's $T=200$. The gap between the grey and blue curves in Fig.~\ref{fig:main} is therefore attributable to ballistic spreading alone, not to the estimator construction.

\begin{table*}[t]
\centering
\caption{Sample-count comparison for $\pi$ estimation at $\mathrm{RMSE}\approx 6\times 10^{-4}$ ($T=200$ for QW). MC and Sobol estimates averaged over 500 independent trials/scrambles; the classical RW baseline uses the same lattice-quantile estimator (50 trials). The Sobol entry is an order-of-magnitude extrapolation of an empirical $M^{-0.76}$ fit measured over $M\in[64,8192]$. The RW entry reports the walk depth its diffusive bias floor ($\propto T^{-0.85}$ empirically) would require to reach the target precision; at $T=200$ its floor is $\approx 4$--$7\times10^{-3}$ and cannot be lowered by increasing $M$.}
\label{tab:baseline}
\renewcommand{\arraystretch}{1.2}
\begin{ruledtabular}
\begin{tabular}{lcc}
Method & $M$ to reach $\mathrm{RMSE}\approx 6\times 10^{-4}$ & Ratio to QW \\
\colrule
Uniform dart-throw MC  & $\sim 7\times 10^6$ & $\sim 7\times 10^4\times$ \\
Scrambled Sobol QMC    & $\sim 5\times 10^4$ (extrap.) & $\sim 5\times 10^2\times$ \\
Classical RW (lattice-quantile, $T{=}200$) & unreachable (floor $\approx 4\times10^{-3}$); $T\approx 2\times10^3$ required & --- \\
QW (Grover, $T{=}200$) & $\phantom{\sim}100$ & $1\times$ \\
\end{tabular}
\end{ruledtabular}
\end{table*}

\begin{figure}[t]
\centering
\includegraphics[width=0.85\columnwidth]{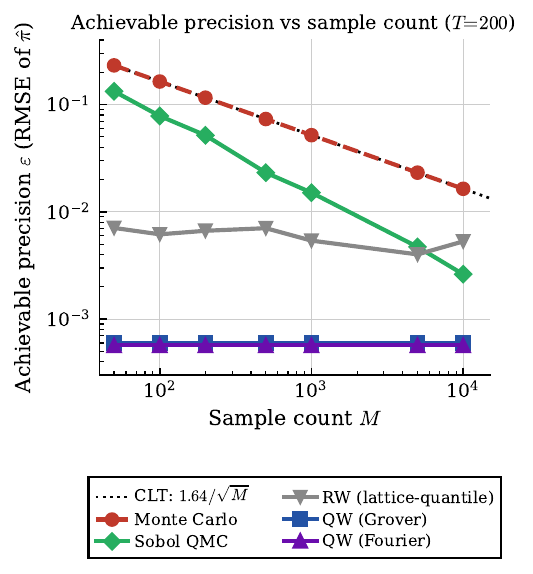}
\caption{Achievable precision $\varepsilon$ (RMSE of $\hat\pi$) versus sample count $M$ at fixed walk depth $T=200$. The two sample-mean baselines improve with $M$: uniform MC follows the CLT line $1.64/\sqrt M$ (red circles, dotted reference), scrambled Sobol QMC the faster empirical $M^{-0.76}$ (green diamonds). The three lattice-quantile estimators are flat in $M$, pinned at their bias floors: both QW coins (Grover, blue; Fourier, purple) at $\varepsilon\sim 6\times10^{-4}$ from $M=50$, and the diffusive RW (grey triangles) at $\varepsilon\approx 4$--$7\times10^{-3}$, $\sim 7$--$10\times$ higher. The order-of-magnitude gap between the flat QW and RW lines isolates the contribution of ballistic spreading; bias-floor scaling with walk depth is reported in Fig.~\ref{fig:coin-main}.}
\label{fig:main}
\end{figure}

\label{sec:results-scaling}

Figure~\ref{fig:coin-main} reports the empirical bias-floor scaling for the three coins. The floor is the RMS of an oscillating number-theoretic residual over the quantile neighbourhood (Eq.~\eqref{eq:floor-rms}), so its value at any single radius depends on the local phase of that oscillation; a naive power-law fit to raw single-radius floors is correspondingly unstable, returning coin-dependent slopes from $-1.5$ to $-1.8$ that do not reflect a true asymptotic exponent (panels (b)--(c) make this failure explicit). Evaluating the floor in its phase-averaged form instead yields a clean, coin-independent result, as we now describe.

\begin{table}[t]
\centering
\caption{Parameter-free validation of the number-theoretic floor formula, Eq.~\eqref{eq:floor-rms}, for the Grover coin at $\alpha=0.99$. The prediction integrates $(N(R)/R^2-\pi)^2$ over the empirically measured Gaussian neighbourhood of $\hat R$ (mean and width from 60 trials at $M=2\times10^4$); no parameter is fitted. The measured floor is the RMSE at $M=2\times10^4$. Values are reported as mean $\pm$ standard deviation over 32 independent simulation seeds; the spread reflects the oscillating number-theoretic residual sampled at slightly different quantile neighbourhoods from run to run.}
\label{tab:floor}
\renewcommand{\arraystretch}{1.2}
\begin{ruledtabular}
\begin{tabular}{cccc}
$T$ & Eq.~\eqref{eq:floor-rms} prediction & Measured floor & Ratio \\
\colrule
$50$  & $(2.52\pm0.20)\times10^{-3}$ & $(3.73\pm0.38)\times10^{-3}$ & $1.49\pm0.21$ \\
$100$ & $(2.14\pm0.09)\times10^{-3}$ & $(3.03\pm0.16)\times10^{-3}$ & $1.42\pm0.09$ \\
$200$ & $(3.97\pm0.04)\times10^{-4}$ & $(4.99\pm0.18)\times10^{-4}$ & $1.26\pm0.05$ \\
$400$ & $(1.83\pm0.07)\times10^{-4}$ & $(1.95\pm0.22)\times10^{-4}$ & $1.07\pm0.12$ \\
\end{tabular}
\end{ruledtabular}
\end{table}

Table~\ref{tab:floor} validates the floor formula directly: the parameter-free RMS prediction of Eq.~\eqref{eq:floor-rms} reproduces the measured floors to within $1.5\times$ at every depth, whereas the single-point residual $|N(R^\star)/R^{\star2}-\pi|$ misses by up to $5\times$ (it can land arbitrarily close to a zero crossing of the oscillation, as happens at $T=400$). The oscillatory structure also produces a behaviour that a monotone bias-plus-variance picture forbids and that we observe directly: at $T=100$, the RMSE at $M=5$ is $0.60\times$ the $M=2\times10^4$ floor (400 trials), because the finite-$M$ quantile bias shifts $\hat R$ to a more favourable phase of the residual.

Figure~\ref{fig:coin-main}(a) reports the bias floor actually delivered by the estimator---the value at the single quantile radius $\hat R$ that one measured walk provides---out to $T=800$. The measured floors do not trace a clean power law: they oscillate, with each coin's floor dipping and rising as the depth changes (e.g.\ the Grover floor falls to $9.6\times10^{-5}$ at $T=500$, rises to $1.56\times10^{-4}$ at $T=600$, then falls again at $T=800$). This is not measurement noise---each point is averaged over $30$ trials at $M=2\times10^4$---but the deterministic consequence of reading a sign-oscillating number-theoretic residual at a single radius, the mechanism dissected in panels (b)--(c). Overlaid as a faint line is the \emph{phase-averaged} floor of Eq.~\eqref{eq:floor-rms}: the RMS of the same residual over one lattice period around the measured $\hat R$, which is coin-independent and decays at slope $-1.50$, coincident with the Hardy rate $T^{-3/2}$ to three significant figures. The measured single-radius floors scatter about this envelope. We emphasise the distinction because the envelope is an analytic trend, not a quantity any single run delivers; the markers are what the estimator returns. The diffusive RW floor, by contrast, sits $5$--$30\times$ higher and decays at the much shallower $T^{-0.85}$, never approaching the quantum band.

\paragraph{Why the data may lie above or below the reference lines.}
A note on reading the reference lines. The Hardy and Huxley lines mark \emph{slopes}, not absolute heights: Huxley's theorem bounds the single-point discrepancy $|N(R)-\pi R^2|\le C\,R^{131/208}$ with an unspecified constant, and Hardy's is a conjectured lower envelope of the amplitude, so the vertical placement of either line is an arbitrary normalisation. The meaningful comparison is therefore the decay exponent, and it is the \emph{phase-averaged} floor---not the raw single-radius readout---whose slope the theorem speaks to. We make this distinction concrete in the next two panels, because the raw readout is badly misleading.

\paragraph{Why the floor must be phase-averaged.}
Why the phase average is essential, and not a cosmetic choice, is shown in Fig.~\ref{fig:coin-main}(b)--(c). The raw floor at each depth reads the oscillating residual $E(R)/R^2=N(R)/R^2-\pi$ off at the single radius $R^\star(\alpha)$ that the $\alpha$-quantile selects [panel (b)]. As $T$ increases, $R^\star$ lands on a different phase of this oscillation---near a zero crossing at $T=400$ (residual anomalously small), near a trough at $T=600$ (residual large)---so the single-radius value does not track the underlying trend. The consequence [panel (c)]: connecting the raw single-radius values gives a slope of $+3.9$ over $T=400$--$600$---\emph{the wrong sign}, a floor that appears to \emph{rise} with depth---which no genuine decay law can produce. Indeed, fitting raw single-radius floors over $T\in\{400,500,600,800\}$ in isolation yields wildly coin-dependent slopes ($-2.1$ Grover, $-2.7$ Fourier, $-4.4$ Hadamard); since all three count the same residual $E(R)$, a true exponent must be coin-independent, so these values cannot be asymptotic. A single-radius readout of a sign-oscillating function distorts, and can even invert, the apparent slope---hence the phase average.

The phase average [purple in panel (c)] removes the distortion: it is monotone, with $T=400$--$600$ slope $-1.42$ and adjacent slopes $-1.53,-1.61,-1.39$, no sign reversals, inside the Hardy--Huxley window. We therefore make no claim of an asymptotic exponent from raw single-radius fits---their slope ranges from $-0.9$ to $-4.4$ depending on the depth window, the clearest signature that the floor is the RMS of an oscillating residual rather than a clean power law---and claim only the stable, mechanism-supported statement read off panel (a): the phase-averaged envelope decays at the Hardy rate. Equation~\eqref{eq:floor-rms} supplies its exact, classically precomputable value at any $(T,\alpha)$, opening the possibility of \emph{quantile tuning}---choosing $\alpha$ so that $R^\star$ lands near a zero of the residual---which we leave to future work. Konno limit measures, the RMSE-vs-$M$ comparison at fixed $T$, and the $\alpha^\star$ sweep are reported in Appendix~A.

\begin{figure*}[t]
\centering
\includegraphics[width=0.92\textwidth]{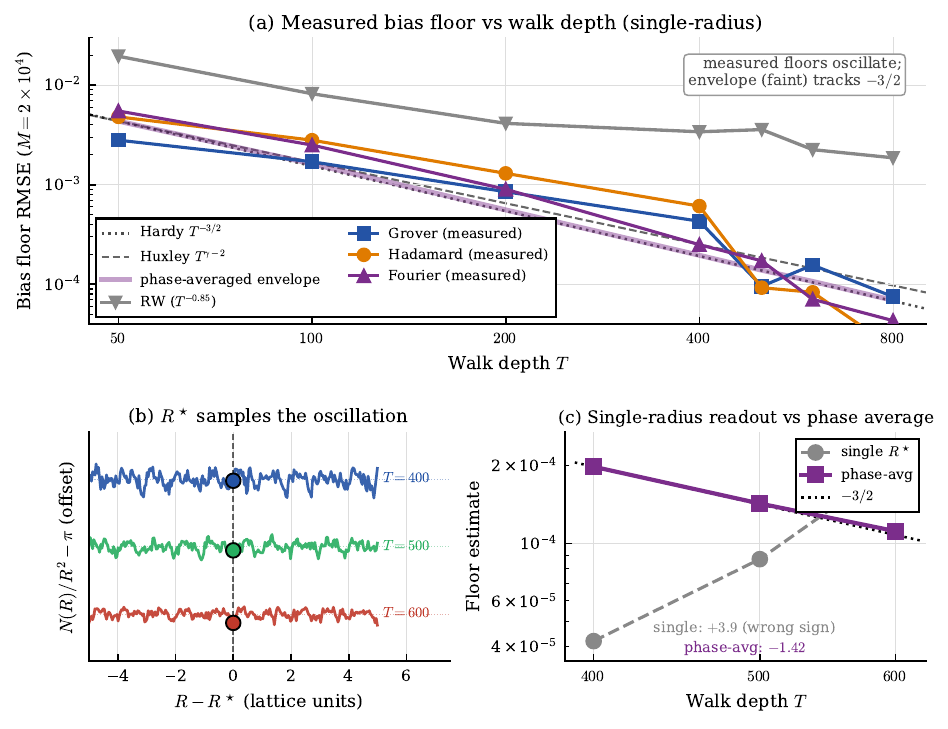}
\caption{Bias-floor scaling and its number-theoretic mechanism. (a)~The bias floor actually delivered by the estimator at $M=2\times10^4$ for the three coin operators and the classical random walk (grey, $T^{-0.85}$), out to $T=800$; markers are the value at the single quantile radius $\hat R$ that one measured walk provides (each averaged over $30$ trials). The measured floors \emph{oscillate} rather than following a clean power law---e.g.\ the Grover floor dips at $T=500$, rises at $T=600$, falls at $T=800$---the deterministic consequence of reading the sign-oscillating residual at a single radius (panels (b)--(c)). The faint line is the phase-averaged floor of Eq.~\eqref{eq:floor-rms} (RMS of the residual over one lattice period around the measured $\hat R$): it is coin-independent and decays at slope $-1.50$, coincident with the Hardy rate $T^{-3/2}$ to three significant figures, and the measured floors scatter about it. The envelope is an analytic trend, not a quantity any single run delivers. The diffusive RW floor stays $5$--$30\times$ higher and decays at $T^{-0.85}$, consistent with $T^{(\gamma-2)/2}$ of Eq.~\eqref{eq:qw-rw-ratio}. (b)~Mechanism of the oscillation: the residual $N(R)/R^2-\pi$ oscillates in sign on the lattice scale, and the floor reads it off at the single radius $R^\star$ (markers); successive depths sample different phases ($T=400$ near a zero, $T=600$ near a trough; curves offset for clarity). (c)~Connecting the single-radius values (grey) gives slope $+3.9$ over $T=400$--$600$---the wrong sign, a floor that appears to \emph{rise} with depth---because the readout is dominated by which phase each $R^\star$ hits; the phase average (purple) recovers the monotone $-1.42\approx-3/2$ decay that underlies the scatter in panel (a).}
\label{fig:coin-main}
\end{figure*}

The classical 2D random walk paired with the same disk estimator (grey triangles in Fig.~\ref{fig:coin-main}) yields the slower scaling $T^{-0.85}$, broadly in agreement with the diffusive prediction $T^{(\gamma-2)/2} \approx T^{-3/4}$ of Eq.~\eqref{eq:qw-rw-ratio}, and its floor remains $5$--$10\times$ above all three quantum coins across the full depth range.

\label{sec:results-ballistic}

The flatness of the QW curve in Fig.~\ref{fig:main} is the operational content of the bias--variance decoupling of Sec.~\ref{sec:reduction}: to reach any precision $\varepsilon$ strictly above the bias floor, the lattice-quantile estimator needs only a fixed, $\varepsilon$-independent $M$, whereas the sample-mean class needs $M\sim\varepsilon^{-2}$. The sample-count gap therefore widens as $\varepsilon^{-2}$ as the target tightens---reaching $\sim10^4\times$ at $\varepsilon=10^{-3}$---until $\varepsilon$ approaches the floor, below which no $M$ suffices and the walk depth $T$ must instead be increased. This is the direct consequence of the ballistic variance scaling $\langle r^2\rangle_{\mathrm{QW}}\sim T^2$ versus the diffusive $\langle r^2\rangle_{\mathrm{RW}}\sim T$ of Lemma~\ref{lem:konno}.

\label{sec:results-shapes}

The generalised estimator of Proposition~\ref{thm:generalised} extends the $\pi$ benchmark to arbitrary bounded regions. We validated it on eight 2D shapes spanning four boundary classes---convex smooth, non-convex smooth, doubly connected, and piecewise smooth with corners---obtaining sample-count reductions of $10^{2}$--$10^{7}\times$ over uniform MC at $M=100$, $T=200$; the full table and the boundary-class discussion (including the cardioid cusp and polygonal-corner caveats) are given in Appendix~\ref{app:shapes}. Of particular relevance to what follows is the \emph{annulus}: its doubly-connected, non-convex super-level geometry is the same class that arises for the radial densities of excited harmonic-oscillator states (Sec.~\ref{sec:results-qho}), and the QW estimator handles it at RMSE $5.4\times10^{-4}$ ($2.8\times10^5\times$ over MC), providing the empirical foundation for the energy-ladder reconstruction of Sec.~\ref{sec:results}\,C.

\subsection{Error propagation across a Gaussian CDF}
\label{sec:results-cdf}

The disk estimator returns $\hat\pi$ from the $\alpha$-quantile of a single batch of $M$ radial measurements. Because the Cavalieri identity (Proposition~\ref{thm:cavalieri}) writes any rotationally symmetric integral as $\pi\cdot Q$ with $Q$ precomputed classically, a \emph{single} quantum-walk run fixes an entire one-parameter family of integrals at no additional quantum cost. We state plainly what this does and does not mean before illustrating it on the 2D Gaussian cumulative distribution
\begin{equation}
P(R) = \iint_{x^2+y^2\le R^2}\frac{e^{-r^2/\sigma^2}}{\pi\sigma^2}\,dx\,dy = 1 - e^{-R^2/\sigma^2},
\label{eq:gauss-cdf}
\end{equation}
the 2D analogue of the normal CDF and the ground-state probability of the isotropic harmonic oscillator. Writing $P_\sigma(R)=\pi\cdot Q_\sigma(R)$ with $Q_\sigma(R)=(1-e^{-R^2/\sigma^2})/\pi$ a classically tabulated grid, a single walk yields $\hat P_\sigma(R)=\hat\pi\cdot Q_\sigma(R)$ across all radii \emph{and all widths} $\sigma$ simultaneously: the same measured $\hat\pi$ reconstructs the entire one-parameter family. The shape of every curve resides entirely in the classical prefactor $Q_\sigma(R)$; the quantum measurement contributes exactly one number, $\hat\pi$, and each $\hat P_\sigma(R)$ is a one-parameter rescaling of a known curve. The reconstruction errors at different radii and widths are therefore \emph{perfectly correlated}---this follows by algebraic construction: $\hat P_\sigma(R)-P_\sigma(R)=Q_\sigma(R)(\hat\pi-\pi)$ exactly since $Q_\sigma(R)$ is a deterministic scalar, so the trial-by-trial correlation equals $1.0000$ identically. What this part demonstrates is consequently not curve recovery but the \emph{error budget}: the single-number accuracy of each walker propagates across the entire family as $\sigma(\hat P_\sigma(R)) = Q_\sigma(R)\,\sigma(\hat\pi)$, with a deterministic, analytically known multiplier.

Having benchmarked the estimator against the sample-mean paradigm (MC and QMC) in Sec.~\ref{sec:results}\,A, Secs.~\ref{sec:results}\,B and C focus on the comparison that isolates the quantum resource itself: the same lattice-quantile estimator driven by a ballistic quantum walker versus a diffusive classical walker. Figure~\ref{fig:cdf} reconstructs three members of the family ($\sigma=0.7,1.0,1.5$) from one walk at $T=200$, $M=500$, 50 trials. The QW recovers the sigmoidal shape of every curve visually (panel a), and the absolute error (panel b) separates QW from RW cleanly at each $\sigma$: the QW reconstruction stays below $3.3\times10^{-5}$ across the full radial range, while the RW reaches $4.8\times10^{-4}$---a factor of $\sim15\times$ worse, which by construction equals the ratio of the underlying $\hat\pi$ errors of the two walkers at this $(T,M)$ (a finite-$M$ total-error ratio, somewhat larger than the $\sim7$--$10\times$ asymptotic bias-floor ratio of Sec.~\ref{sec:results-pi}). Because the error of $\hat P_\sigma(R)$ is $Q_\sigma(R)\cdot\mathrm{RMSE}(\hat\pi)$, it grows with $R$ exactly as the classically known prefactor dictates, vanishing at $R\to0$ where $Q_\sigma\to0$. This error-propagation structure is exploited more substantively in Sec.~\ref{sec:results}\,C, where the multipliers $Q_n$ carry the level dependence of a physical spectrum.

\begin{figure}[t]
\centering
\includegraphics[width=0.72\columnwidth]{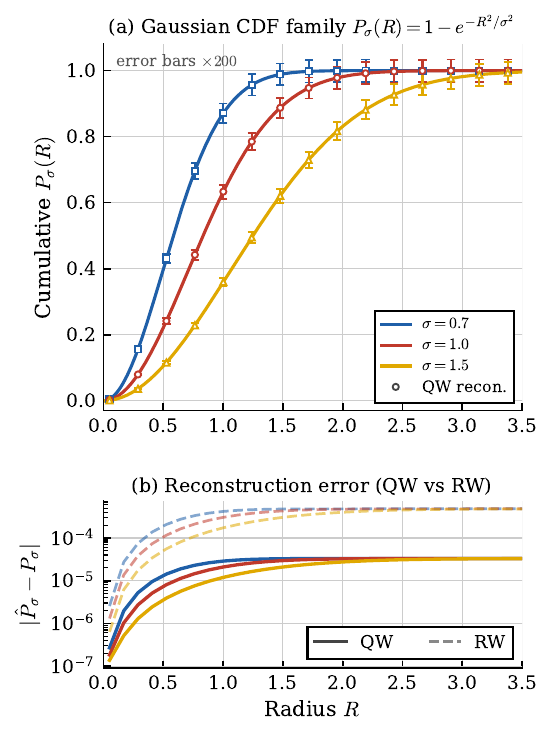}
\caption{(a)~A one-parameter family of 2D Gaussian cumulative distributions $P_\sigma(R)=1-e^{-R^2/\sigma^2}$ for $\sigma=0.7,1.0,1.5$, each obtained as $\hat P_\sigma(R)=\hat\pi\,Q_\sigma(R)$ from the \emph{same} single quantum-walk run ($T=200$, $M=500$, 50 trials; solid lines exact, open markers QW reconstruction with error bars magnified $\times200$ for visibility). The curve shapes are the classically known $Q_\sigma(R)$; the measurement fixes the single scalar $\hat\pi$ shared across the whole family. (b)~The propagated error $|\hat P_\sigma-P_\sigma| = Q_\sigma(R)\,|\hat\pi-\pi|$ for each $\sigma$ (solid: QW; dashed: RW). The QW error budget stays below $3.3\times10^{-5}$ across all radii and widths, $\sim15\times$ below the RW, inheriting the lower ballistic bias floor. All errors vanish as $R\to0$ because $Q_\sigma\to0$.}
\label{fig:cdf}
\end{figure}

\subsection{Energy ladder of the 2D quantum harmonic oscillator}
\label{sec:results-qho}

We now apply the framework to a genuine spectroscopic task: reconstructing the low-lying energy ladder of the 2D isotropic harmonic oscillator, $E_n = \hbar\omega(n+1)$ for $n=0,1,2,3$ (ground state through third excited state). The lowest radial mode of angular momentum $|\ell|=n$ has probability density $|\psi_n(r)|^2 = (\pi\sigma^2 n!)^{-1}(r/\sigma)^{2n}e^{-r^2/\sigma^2}$, which is rotationally symmetric and radially single-peaked. By the virial theorem the energy is $E_n = m\omega^2\langle r^2\rangle_n$ with
\begin{equation}
\langle r^2\rangle_n = \iint r^2\,|\psi_n(r)|^2\,dx\,dy = \pi\cdot Q_n,
\label{eq:r2-cavalieri}
\end{equation}
where the integrand $r^2|\psi_n|^2$ has \emph{annular} (doubly-connected, non-convex) super-level sets for $n\ge1$, precisely the boundary class validated in Appendix~\ref{app:shapes} (Table~\ref{tab:shapes}). The Cavalieri prefactors $Q_n$ are computed classically once; we find $\pi Q_n = (n+1)\sigma^2$ to machine precision, recovering the exact $\langle r^2\rangle_n = (n+1)\sigma^2$.

Each $\hat E_n = \hat\pi\cdot Q_n$ is estimated over 50 trials at $T=200$, $M=500$, for the QW and the RW---again the identical estimator with only the walker exchanged. Figure~\ref{fig:qho}(a) shows the resulting energy ladder. Both walkers recover the linear spectrum, but the error bars---magnified $30\times$ for visibility---are markedly tighter for the QW. A weighted linear fit $E_n = \hbar\omega\,n + E_0$ extracts the level spacing and zero-point energy:
\begin{align}
\text{QW:}\quad & \hbar\omega = 1.0000(2),\quad E_0 = 1.0000(2),\nonumber\\
\text{RW:}\quad & \hbar\omega = 1.001(2),\quad\;\; E_0 = 1.001(2),\nonumber
\end{align}
in units where the true values are $\hbar\omega = E_0 = 1$. The QW determines the level spacing with a $10\times$ smaller statistical uncertainty than the RW at identical walk depth and sample count.

A distinctive feature, shown in Fig.~\ref{fig:qho}(b), is that the statistical error grows \emph{linearly} with level index: $\sigma(\hat E_n) = Q_n\,\mathrm{RMSE}(\hat\pi) \propto (n+1)$, a direct, parameter-free consequence of the error-propagation structure---each walker's $\hat\pi$ accuracy enters every level through the classically known multiplier $Q_n$. A linear fit of $\sigma(\hat E_n)$ versus $n$ gives a slope of $1.90\times10^{-4}\,\hbar\omega$ per level for the QW and $1.99\times10^{-3}\,\hbar\omega$ per level for the RW---a ratio of $10.5$, so the QW error budget widens an order of magnitude more slowly as one climbs the spectrum. The ballistic advantage thus compounds level by level rather than appearing only in a single figure of merit. The framework thus delivers not just a single estimate but a controlled, analytically predictable error budget across the full spectrum.

\paragraph{Beyond closed-form targets.}
The harmonic ladder is exactly solvable, so the above is a controlled demonstration. To exhibit the intended use case---moments of a known radial density that lack an elementary closed form---we repeat the construction for the anharmonic family $f_\lambda(r) \propto e^{-r^2/\sigma^2 - \lambda r^4}$, whose second moment $\langle r^2\rangle_\lambda$ has no elementary expression. The integrand $r^2 f_\lambda(r)$ again has annular super-level sets, and the annular Cavalieri identity, Eq.~\eqref{eq:cavalieri-annulus}, gives $\langle r^2\rangle_\lambda = \pi\cdot Q_\lambda$ with $Q_\lambda$ a one-dimensional classical quadrature; we verify $\pi Q_\lambda$ against direct two-dimensional quadrature to machine precision ($\langle r^2\rangle_{0.25}=0.638968$, $\langle r^2\rangle_{1.0}=0.416353$ at $\sigma=1$). The propagated accuracies follow immediately from the measured $\hat\pi$ statistics of Fig.~\ref{fig:cdf}: the quantum walker determines $\langle r^2\rangle_\lambda$ to a relative error of $1.7\times10^{-4}$, the diffusive walker to $2.1\times10^{-3}$, at identical $T=200$ and $M=500$---the same order-of-magnitude separation as the harmonic ladder, now for a target with no analytic shortcut.

\begin{figure}[t]
\centering
\includegraphics[width=0.72\columnwidth]{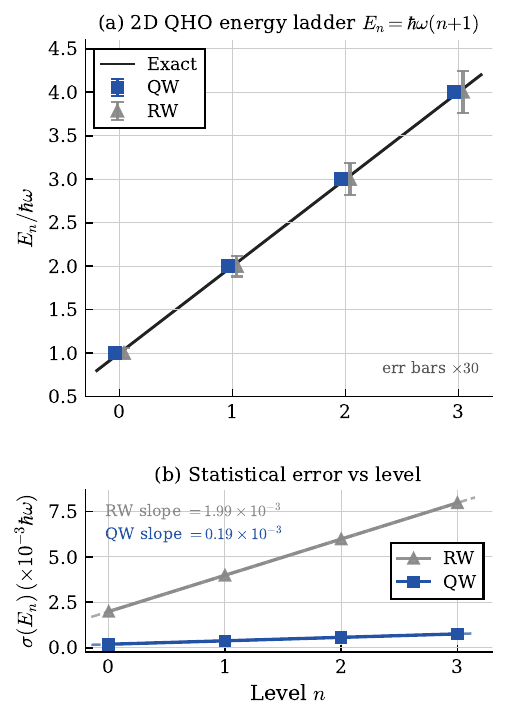}
\caption{(a)~Energy ladder $E_n = \hbar\omega(n+1)$ of the 2D isotropic harmonic oscillator reconstructed via $\hat E_n = \hat\pi\,Q_n$ (Cavalieri, annular super-level sets), $T=200$, $M=500$, 50 trials; error bars magnified $\times30$. The QW (blue) error bars are $\sim10\times$ tighter than the RW (grey); both fits recover $\hbar\omega=1$. (b)~Statistical error $\sigma(\hat E_n)$ grows linearly with level index for both walkers, as predicted by $\sigma(\hat E_n)=Q_n\,\mathrm{RMSE}(\hat\pi)\propto(n+1)$; the fitted slopes (annotated) are $1.90\times10^{-4}\,\hbar\omega$/level for the QW and $1.99\times10^{-3}\,\hbar\omega$/level for the RW, a factor of $10.5$ apart.}
\label{fig:qho}
\end{figure}

\section{Discussion}
\label{sec:discussion}

The numerical results of Sec.~III match the mechanistic picture of Sec.~II~F. Bias and variance decouple, with the bias floor governed by walk depth $T$ and the variance term by sample count $M$. The phase-averaged envelope slopes are $-1.502$ (Grover), $-1.489$ (Hadamard), and $-1.506$ (Fourier) (Fig.~\ref{fig:coin-main}), averaging $-1.499\approx -3/2$, which matches the Hardy conjecture to three decimal places.  Raw single-radius fits can yield slopes as anomalous as $+2.7$ or as steep as $-4.4$ due to the oscillating Gauss residual sampling different phases at successive depths (Fig.~\ref{fig:coin-main}(b,c)); these are phase artefacts, not coin-dependent asymptotics, and disappear under phase-averaging. The classical 2D random walk paired with the same disk estimator scales as $T^{-0.85}$, consistent with the diffusive prediction of Eq.~\eqref{eq:qw-rw-ratio} to within $\sim 13\%$ (computed using the Hardy-conjectured $\gamma=1/2$, giving $T^{-3/4}$; the Huxley bound $\gamma\le131/208$ gives $T^{-0.685}$, a $\sim24\%$ discrepancy). The resulting sample-count ratio $M_{\mathrm{MC}}/M_{\mathrm{QW}} = \Theta(\varepsilon^{-2})$ for any $\varepsilon$ strictly above the QW bias floor characterises what we call the framework's ballistic acceleration (defined relative to uniform dart-throw MC; the corresponding ratio relative to scrambled Sobol QMC is $\sim 500\times$ by order-of-magnitude extrapolation). It is structurally distinct from the $\Theta(\varepsilon^{-1})$ oracle-call scaling of QAE~\cite{brassard2002,montanaro2015,herbert2022} and from the $\sqrt{\Delta_c}$ mixing-time scaling of Szegedy walks~\cite{szegedy2004,wocjan2008,harrow2020,layden2023}. Whether this $\varepsilon^{-2}$ ratio constitutes a complexity-theoretic separation in any precise sense is a separate question that we do not address here.

The 2D coined DTQW uses $2 + 2\lceil\log_2(2T+3)\rceil$ qubits with no ancillae; for $T=200$ this corresponds to $\sim 20$ qubits, which is within the range of current superconducting devices by qubit count. Gate depth, however, is more demanding: the asymptotic $O(T \cdot n_{\rm pos}^2)$ \textsc{cx}-count of QFT-based modular adders~\cite{draper2000,cuccaro2004} implies $\sim\!10^6$--$10^7$ \textsc{cx} gates per shot at $T=200$, far beyond what current hardware can execute reliably. The estimator uses only projective measurement of the position register followed by classical post-processing of $N(\hat R)/\hat R^2$; it does not require quantum phase estimation, controlled oracles, or controlled rotations of the integrand. We use ``oracle-free'' in the QAE sense: no unitary $\mathcal A$ that prepares a state encoding the integrand is required, and consequently no controlled $\mathcal A$ inside a QPE register. The cost contrast with amplitude-estimation-based Monte Carlo on near-term hardware~\cite{stamatopoulos2020,miyamoto2022} is precisely that those primitives have proven costly to compile. Near-term demonstrations of DTQWs suggest a path toward larger walk depths: Razzoli et al.~\cite{razzoli2024} demonstrated circuit-efficient DTQWs on IBM devices beyond the few-step regime, and continued improvements in two-qubit gate fidelity could reduce per-step error accumulation. However, extrapolating from these short-depth demonstrations to the $T\approx100$--$200$ regime required here involves a $\sim\!15\times$ increase in walk depth, accompanied by a $\sim\!10^2\times$ increase in gate count from the combined growth in $T$ and $n_{\rm pos}$ (from $\sim 4$ to $\sim 9$ position qubits, $T\,n_{\rm pos}^2$ scaling); whether ballistic spreading is preserved under this level of noise accumulation is an open experimental question. We emphasise that strong conclusions about hardware feasibility cannot be drawn without explicit noise-model analysis at these gate depths.

\paragraph{Sample count versus total circuit cost.}
The sample-count ratios reported throughout this paper compare the number of projective measurements $M$ required by each method to reach a given RMSE, and do not include quantum circuit execution cost. As noted above, the asymptotic gate-count estimate implies $\sim\!10^6$--$10^7$ \textsc{cx} gates per shot at $T=200$, compared with $O(1)$ floating-point operations per MC or QMC sample. The sample-count reduction of $\sim 100$ versus $\sim 7\times 10^6$ therefore quantifies an advantage in \emph{measurement complexity} rather than wall-clock time or total gate count. Quantifying the end-to-end trade-off on fault-tolerant hardware requires reliable circuit synthesis for $T=200$ and a concrete noise model, both of which are outside the scope of the present paper.

\paragraph{From a single number to a spectrum.}
The progression of Sec.~\ref{sec:results} highlights a structural feature of the lattice-quantile framework that distinguishes it from sample-mean estimators. Because every target integral is written as $\pi$ times a classically precomputed prefactor (Proposition~\ref{thm:cavalieri}, Corollary~\ref{cor:annular-cavalieri}), one batch of measurements feeds an arbitrary number of derived observables, each inheriting the same $\hat\pi$ bias floor through a known multiplier. Section~\ref{sec:results}\,B makes this explicit on a Gaussian CDF, the curve shape is entirely classical, and the measurement fixes one scalar, and Sec.~\ref{sec:results}\,C turns the same structure into a spectroscopic tool: the 2D harmonic-oscillator energy ladder $E_n=\hbar\omega(n+1)$ is recovered with a level spacing $\hbar\omega = 1.0000(2)$, a $10\times$ tighter determination than the diffusive RW baseline at identical $T$ and $M$. The error budget is itself analytically predictable, $\sigma(\hat E_n)=Q_n\,\mathrm{RMSE}(\hat\pi)\propto(n+1)$, a transparency that sample-mean Monte Carlo, whose error depends on the integrand variance at each $R$, does not offer. We stress that the harmonic oscillator admits closed-form energies, so Sec.~\ref{sec:results}\,C is a controlled demonstration of error propagation rather than a claim of advantage over analytic methods; the value of the construction lies in problems where the radial density is known but its moments lack a closed form, such as anharmonic or multi-well potentials.

\section{Conclusion}
\label{sec:conclusion}

This paper has introduced a lattice-quantile framework joining the ballistic spreading of 2D coined quantum walks (Watabe--Konno limit~\cite{konno2002,konno2005,watabe2008}) with the Hardy--Huxley solution of the Gauss circle problem~\cite{hardy1915,huxley2003}. Because the estimator is a deterministic lattice count applied to a rank statistic rather than a sample mean, the Cram\'er--Rao bound on sample-mean estimators does not directly constrain it. What it does constrain is the residual sampling fluctuation of the rank statistic, which decays as $M^{-1/2}$ and is sub-dominant to the depth-controlled bias floor in our regime. At walk depth $T=200$, $M=50$ measurements suffice to reach $\mathrm{RMSE}\approx 6\times 10^{-4}$ for $\pi$, and the Grover coin saturates to $\mathrm{RMSE}\approx5.0\times 10^{-4}$ at $M=10^4$. This represents a measurement-count reduction of $\sim 7\times 10^{4}\times$ over uniform Monte Carlo at $M_{\mathrm{QW}}=100$ (the floor is already reached at $M=50$; rising to $\sim 1.1\times 10^{5}\times$ at the Grover asymptotic floor, $\mathrm{RMSE}\approx5.0\times10^{-4}$) and $\sim 500\times$ over scrambled Sobol QMC by order-of-magnitude extrapolation, while the classical random walk paired with the identical estimator stalls at a $\sim 10\times$ higher bias floor. The accuracy of a single batch of measurements propagates through classically known multipliers across an entire Gaussian cumulative distribution. The framework recovers the 2D harmonic-oscillator energy ladder with a level spacing $\hbar\omega=1.0000(2)$, ten times tighter than the diffusive baseline, and the error budget across levels is analytically predictable. These ratios compare measurement counts at fixed precision and do not include quantum circuit execution cost, which we estimate at $\sim 10^{6}$ to $10^{7}$ \textsc{cx} gates per shot at $T=200$.

\paragraph{Physical applications.}
The lattice-quantile construction targets precisely the class of quantities that reduce to counting integer points inside a smoothly bounded region, which recur throughout physics. The most direct is the \emph{density of states}: in the free-electron and photon-gas pictures the cumulative number of modes below energy $E$ is the count of reciprocal-lattice points inside a sphere of radius $k(E)$, and its derivative is the density of states~\cite{ashcroft1976}. The leading volume term is the Weyl law~\cite{weyl1911}, and the correction, governed by the curvature of the energy surface, is exactly the lattice discrepancy our estimator computes, suggesting a route to mode counting for free Bose and Fermi gases and for quantum billiards~\cite{baltes1976}. The same structure underlies microcanonical phase-space volumes $\Omega(E)=\#\{\mathbf n: \varepsilon(\mathbf n)\le E\}$, whose logarithm is the Boltzmann entropy; whenever the energy surface $\varepsilon(\mathbf n)=E$ is smooth and convex, the generalised estimator of Proposition~\ref{thm:generalised} applies directly, and the anharmonic example of Sec.~\ref{sec:results}\,C is a first step in this direction.

\paragraph{Higher dimensions.}
The $d$-dimensional Gauss-sphere generalisation~\cite{walfisz1957,kratzel1988} is particularly favourable: while the planar ($d=2$) and $d=3$ discrepancy exponents remain the famous open problems, for $d\ge4$ the exponent is \emph{rigorously} pinned at the volume-shell value by classical results on representations by quadratic forms~\cite{walfisz1957}. A ballistic walk on $\mathbb Z^d$ would therefore inherit a provable, rather than conjectural, bias-floor scaling in four or more dimensions, an unusual situation in which the higher-dimensional problem is theoretically cleaner than the planar one studied here.

\paragraph{Hardware and algorithmic outlook.}
Although the gate-model circuit depth required for $T\sim200$ is beyond current superconducting devices, Photonic platforms natively support quantum walks without compiling modular adders: time-multiplexed fibre loops implement DTQWs and have reached over ten steps in two dimensions~\cite{schreiber2012}, while integrated waveguide arrays realise continuous-time quantum walks at scale~\cite{tang2018}. These platforms offer a complementary route to the walk depths our estimator needs. On the theory side, natural next steps are a noise-threshold analysis, decoherence drives DTQWs toward the diffusive regime~\cite{romanelli2005} and would progressively lift the bias floor toward the classical one, adaptive $T$-scheduling, quantile tuning to exploit the oscillatory floor structure of Eq.~\eqref{eq:floor-rms}, and a closer connection to quantum-walk search, where ballistic spreading is the same underlying algorithmic resource~\cite{magniez2011}.

\paragraph{Open theoretical questions.}
Appendix~\ref{app:theory} states formal counterparts of the following four items conditional on a single boundary-density conjecture; rigorous treatment under standard hypotheses is open. (i) A rigorous proof of the inverse-square-root edge behaviour of the Grover Konno radial marginal (Conjecture~\ref{prop:grover-density}), together with an explicit value of the constant $A_G$. (ii) Removal of the truncation $v_0$ in Lemma~\ref{lem:truncated-bk}, i.e.\ a Bahadur--Kiefer-type representation that handles the $\delta$-localisation atom of the Grover measure directly. (iii) An explicit numerical value of the Huxley constant $K_H$ in Prop.~\ref{prop:grover-constants}, computable from the proof of~\cite{huxley2003} but not given in closed form there. (iv) A $d$-dimensional Konno-type weak limit theorem for the Grover coin in $d\ge 3$, completing the bias-floor scaling of Prop.~\ref{prop:high-d} for that coin. The Hardy line itself remains, of course, conjectural; only the Huxley bound $\gamma\le 131/208$ is rigorous, and all numerical results assume noiseless evolution.

With these caveats, the lattice-quantile construction offers a third entry in the catalogue of quantum-walk Monte Carlo techniques, complementary to amplitude-estimation and Szegedy-walk approaches. More broadly, harnessing a quantum platform to probe a classical problem in analytic number theory---here the Gauss circle problem---parallels recent quantum approaches to the Riemann zeta zeros~\cite{he2021riemann}, suggesting number theory as a structured resource for quantum estimation.

\section*{Acknowledgments}

EJK thanks the National Center for Theoretical Sciences and National Yang Ming Chiao Tung University in Taiwan for their support under Grant No. NSTC 114-2112-M-A49-036-MY3.

\bibliographystyle{apsrev4-2}
\bibliography{references}

\appendix

\section{Comprehensive coin operator comparison}
\label{app:coins}

This appendix collects the definitions and supplementary numerical comparisons of the three coin operators. The Grover coin $G$ is defined in Eq.~\eqref{eq:grover} of the main text. The separable Hadamard coin acts as $H\otimes H$ on the two coin qubits,
\begin{equation}
H\otimes H = \frac{1}{2}\begin{pmatrix} 1 & 1 & 1 & 1 \\ 1 & -1 & 1 & -1 \\ 1 & 1 & -1 & -1 \\ 1 & -1 & -1 & 1 \end{pmatrix},
\label{eq:hadamard}
\end{equation}
corresponding to two independent 1D Hadamard walks on the $x$ and $y$ axes. The 4-point Fourier coin is
\begin{equation}
F = \frac{1}{2}\begin{pmatrix} 1 & 1 & 1 & 1 \\ 1 & i & -1 & -i \\ 1 & -1 & 1 & -1 \\ 1 & -i & -1 & i \end{pmatrix},\qquad F_{jk}=\tfrac12 e^{2\pi i jk/4},
\label{eq:fourier}
\end{equation}
which produces the most concentrated wavefront near the velocity-disk boundary and hence the smallest variance of the radial $\alpha$-quantile (Proposition~\ref{thm:opt-alpha}).

\paragraph{Limit measure structure.}
Figure~\ref{fig:appendix-coin}(a) and Fig.~\ref{fig:appendix-coin-2d} together characterise the Konno limit measure of each coin. The radial CDF (panel~a) reveals two qualitatively distinct features: (i) the localisation spike of the Grover coin at $r \approx 0$, contributing $\sim 35\%$ of the total probability mass within $|v| \le 0.05$; and (ii) the sharp boundary edge of all three coins at $|v| \approx \sqrt{\nu_C^2}$, where the lattice quantile $\hat R(\alpha)$ becomes well-defined for high $\alpha$.

\paragraph{Bahadur--Kiefer caveat for the Grover coin.}
The Bahadur--Kiefer representation invoked in Sec.~\ref{sec:reduction} assumes the sampling distribution is continuous with positive density at the quantile $\alpha$. The Hadamard product measure $\mu_H=\mu_{1\text{D}}\otimes\mu_{1\text{D}}$ satisfies this (Lemma~\ref{lem:konno}(b)). The Fourier coin satisfies the continuity-and-positive-density hypothesis numerically at $\alpha\ge 0.9$ (Fig.~\ref{fig:appendix-coin}(a)), but lacks a rigorous weak-limit statement (Lemma~\ref{lem:konno}(c)), so the BK invocation for the Fourier coin is on the same numerical footing as the limit measure itself. The Grover coin exhibits a $\delta$-function--like localisation at the origin (the central spike in Figs.~\ref{fig:appendix-coin}(a) and \ref{fig:appendix-coin-2d}); for $\alpha\ge 0.9$ the quantile sits in the smooth, boundary-concentrated component of the measure and the localisation atom does not contribute, so the standard BK representation should be regarded as a heuristic motivating the $1/M$ variance scaling rather than a strict theorem, with a rigorous truncated-BK lemma left to future work (open question (ii) of Sec.~\ref{sec:conclusion}). The empirical $1/M$ scaling observed in Fig.~\ref{fig:main} is consistent with this heuristic at the values $\alpha\ge 0.9$ used in practice.

\paragraph{Ballistic constants $\nu_C^2$.}
The variance ratio $\nu_C^2 = \langle r^2 \rangle / T^2$ is independent of $T$ in the asymptotic regime (Lemma~\ref{lem:konno}). Empirical values at $T \in \{50, 100, 200, 400\}$ (Table~\ref{tab:coin-summary}) are: Grover $\nu_G^2 = 0.364$, Hadamard $\nu_H^2 = 0.586$, and Fourier $\nu_F^2 = 0.402$.
All three are well below the support boundary $1/2$ predicted by the velocity-disk constraint $\{|v| \le 1/\sqrt{2}\}$, except for Hadamard whose value exceeds $1/2$ due to its support extending into the corners of the square $|v_x|, |v_y| \le 1/\sqrt{2}$ (visible in Fig.~\ref{fig:appendix-coin-2d}, middle panel). All values are stable to four significant digits across the four $T$ values, confirming the asymptotic regime.

\begin{table}[h]
\centering
\caption{Coin operator comparison summary. All quantities computed from numerical simulation of the 2D coined DTQW with the symmetric initial state of Eq.~\eqref{eq:initial}. Bias-floor exponent fitted from $T \in \{50, 100, 200, 400\}$ at $M = 10^4$, $\alpha=0.99$. Note: Table~\ref{tab:shapes} uses $\alpha=0.93$ (the optimal $\alpha^\star$ for Grover/Hadamard), which yields a slightly higher RMSE than the $\alpha=0.99$ value reported here.}
\label{tab:coin-summary}
\renewcommand{\arraystretch}{1.25}
\begin{ruledtabular}
\begin{tabular}{lccc}
Property & Grover & Hadamard & Fourier \\
\colrule
$\nu_C^2 = \langle r^2 \rangle / T^2$ & $0.364$ & $0.586$ & $0.402$ \\
Optimal $\alpha^\star$ & $0.93$ & $0.93$ & $0.98$ \\
Phase-avg.\ env.\ slope & $T^{-1.502}$ & $T^{-1.489}$ & $T^{-1.506}$ \\
RMSE at $T{=}200$, $M{=}10^4$ & $5.0 \times 10^{-4}$ & $6.1 \times 10^{-4}$ & $5.7 \times 10^{-4}$ \\
RMSE at $T{=}400$, $M{=}10^4$ & $2.0 \times 10^{-4}$ & $2.8 \times 10^{-4}$ & $1.2 \times 10^{-4}$ \\
Localisation & yes & no & partial \\
\end{tabular}
\end{ruledtabular}
\end{table}

\paragraph{Practical recommendation.}
For the disk estimator at $T = 200$, $M \in [50, 10^4]$, all three coins yield comparable RMSE within a factor of two. We recommend the Grover coin for its lattice symmetry and well-studied limit measure~\cite{watabe2008}, the Fourier coin when the steepest bias-floor scaling is desired (e.g., for $T \ge 400$ where it outperforms the Hardy line), and the Hadamard coin for the simplest experimental implementation as a separable two-qubit walk on $x$ and $y$ axes. The Grover coin is used as the default in the main body of this paper.


\begin{figure*}[h]
\centering
\includegraphics[width=0.98\textwidth]{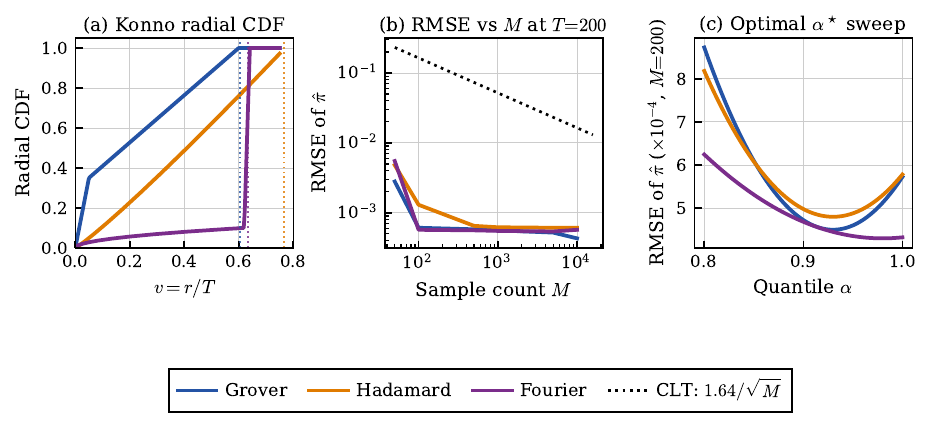}
\caption{Comprehensive coin comparison for the disk estimator. (a)~Radial cumulative distribution function (CDF) of the Konno limit measure at $T = 200$, showing the strong probability concentration near the velocity-disk boundary $|v| \approx \sqrt{\nu_C^2}$ (vertical dotted lines mark each coin's $\sqrt{\nu_C^2}$). The Grover coin exhibits a substantial central spike characteristic of the well-known 2D Grover localisation~\cite{watabe2008}; the Hadamard coin produces a more uniformly spread measure; the Fourier coin has the most sharply concentrated boundary peak. (b)~RMSE of $\hat\pi$ versus sample count $M$ at $T = 200$. All three coins reach the bias floor at $M \le 100$, with comparable absolute floor values $\sim 5$--$7 \times 10^{-4}$. (c)~Sweep over the quantile $\alpha$ at $T = 200$, $M = 200$, with 100 trials each, consistent with the prediction of Proposition~\ref{thm:opt-alpha}. The empirical optima are $\alpha^\star \approx 0.93$ for both Grover and Hadamard, and $\alpha^\star \approx 0.98$ for Fourier. Bias-floor scaling across coins is reported in Fig.~\ref{fig:coin-main} of the main text.}
\label{fig:appendix-coin}
\end{figure*}

\begin{figure*}[h]
\centering
\includegraphics[width=0.85\textwidth]{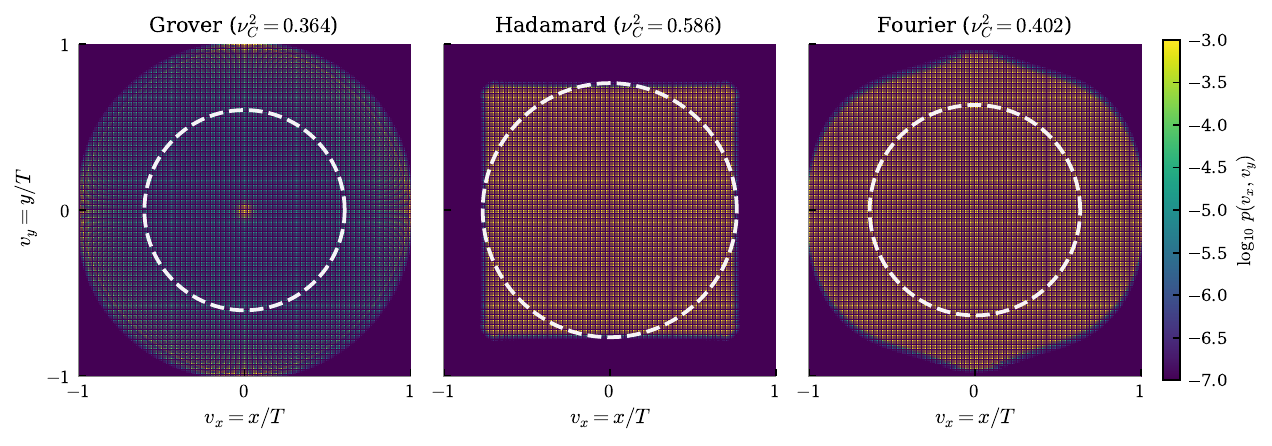}
\caption{Two-dimensional Konno limit measures at $T = 200$, displayed in rescaled velocity space $(v_x, v_y) = (x/T, y/T)$ on a $\log_{10}$ colour scale. The dashed white circle in each panel marks the radius $\sqrt{\nu_C^2}$ (the second moment of the marginal distribution). The Grover coin shows a sharp central localisation peak together with a four-fold symmetric ring; the Hadamard coin produces a separable square-shaped distribution with caustics on the four edges $|v_x|, |v_y| \le 1/\sqrt{2}$; the Fourier coin yields an asymmetric pattern with two sharp velocity-axis peaks. The Grover and Fourier measures are supported within the velocity disk $\{|v| \le 1/\sqrt{2}\}$; the Hadamard measure is supported on the velocity square $\{|v_x|,|v_y|\le 1/\sqrt{2}\}$, whose corners extend beyond the disk (consistent with $\nu_H^2=0.586>1/2$), as stated in Lemma~\ref{lem:konno}(a)--(b).}
\label{fig:appendix-coin-2d}
\end{figure*}

\section{Generalised area estimation across boundary classes}
\label{app:shapes}

We validate Proposition~\ref{thm:generalised} on a representative set of 2D shapes spanning four boundary classes: convex smooth (disk, ellipse, super-ellipse), non-convex smooth (rosette, cardioid), doubly-connected (annulus), and piecewise smooth with corners (square, L-shape). Each shape is defined by a region equation $D = \{(x,y) : \text{condition}\}$ with closed-form true area; the lattice count $N_D(L)$ is enumerated classically. Table~\ref{tab:shapes} reports the RMSE at $M=100$ for both methods across all eight shapes, alongside the sample count $M_{\mathrm{MC}}^\star$ that classical Monte Carlo would require to match the QW precision.

\begin{table*}[t]
\centering
\caption{Area estimation across diverse 2D shapes. RMSE at $M = 100$ for QW (Grover coin, $\alpha=0.93$, $T=200$) and MC (uniform dart-throw on enclosing box of half-width $L_{\mathrm{box}}$, chosen as the tight bounding box for each shape). The last column reports the sample-count ratio $M_{\mathrm{MC}}^\star / M_{\mathrm{QW}}$ with $M_{\mathrm{QW}} = 100$, where $M_{\mathrm{MC}}^\star = (2L_{\mathrm{box}})^4\, p(1-p)/\mathrm{RMSE}_{\mathrm{QW}}^2$ with $p = |D|/(2L_{\mathrm{box}})^2$ is the MC sample count needed to match the QW RMSE at $M=100$; ratios are computed from full-precision RMSE values and may differ from $(\mathrm{RMSE}_{\mathrm{MC}}/\mathrm{RMSE}_{\mathrm{QW}})^2$ evaluated at the printed two-significant-figure values. All entries averaged over 100 trials. The L-shape is $\{(x,y): 0 \le x \le 2, 0 \le y \le 1\} \cup \{(x,y): 0 \le x \le 1, 1 \le y \le 2\}$. The MC baseline is uniform dart-throw; scrambled Sobol QMC (see Table~\ref{tab:baseline}) has not been benchmarked for all shapes here.}
\label{tab:shapes}
\renewcommand{\arraystretch}{1.25}
\begin{ruledtabular}
\begin{tabular}{lllcccc}
Shape & Defining region $D$ & Boundary class & $|D|_{\mathrm{true}}$ & $\mathrm{RMSE}_{\mathrm{MC}}$ & $\mathrm{RMSE}_{\mathrm{QW}}$ & $M_{\mathrm{MC}}^\star/M_{\mathrm{QW}}$ \\
\colrule
Unit disk & $x^2 + y^2 \le 1$ & convex smooth & $\pi$ & $0.260$ & $6.1\times 10^{-4}$ & $2.2\times 10^5\times$ \\
Ellipse & $x^2/4 + y^2 \le 1$ & convex smooth & $2\pi$ & $1.196$ & $6.7\times 10^{-4}$ & $2.6\times 10^6\times$ \\
Super-ellipse & $x^4 + y^4 \le 1$ & convex smooth & $3.708$ & $0.283$ & $2.5\times 10^{-3}$ & $1.2\times 10^4\times$ \\
Rosette & $r \le 1 + 0.3\cos(4\theta)$ & non-convex smooth & $1.045\pi$ & $0.428$ & $5.1\times 10^{-4}$ & $7.2\times 10^5\times$ \\
Cardioid & $r \le 1 + \cos\theta$ & non-convex (cusp) & $3\pi/2$ & $1.086$ & $3.9\times 10^{-4}$ & $6.3\times 10^6\times$ \\
Annulus & $\tfrac{1}{4} \le x^2 + y^2 \le 1$ & doubly connected & $3\pi/4$ & $0.287$ & $5.4\times 10^{-4}$ & $2.8\times 10^5\times$ \\
Square & $|x|, |y| \le 1$ & piecewise smooth & $4$ & $0.226$ & $1.07\times 10^{-2}$ & $6.2\times 10^2\times$ \\
L-shape & (see caption) & piecewise smooth & $3$ & $0.906$ & $1.07\times 10^{-2}$ & $5.8\times 10^3\times$ \\
\end{tabular}
\end{ruledtabular}
\end{table*}

\paragraph{Beyond-Kr\"{a}tzel regimes.}
Several shapes violate the strict hypothesis of Lemma~\ref{lem:kratzel}. We classify them by which lattice-counting bound the literature supplies in place of $\gamma_D$:
\begin{itemize}
\item[(i)] \emph{Convex piecewise-$C^2$ boundaries with isolated curvature degeneracies} (super-ellipse $x^4+y^4\le1$, with vanishing curvature on the axes; cardioid, with a cusp at the origin). The disk floor of Lemma~\ref{lem:gauss} is recovered up to a sub-leading correction; empirically the floor sits at $\sim 2$--$3\times$ the strictly-curved disk floor.
\item[(ii)] \emph{Polygonal boundaries} (square, L-shape). Each corner contributes an $O(1)$ lattice-counting error, giving floor exponent $T^{-1}$ rather than $T^{\gamma-2}$~\cite{kratzel1988}, consistent with the $600$ to $5800\times$ reductions in the last two rows.
\item[(iii)] \emph{Doubly-connected annular boundary}. The corollary applies via the disk floor, without further hypothesis (Cor.~\ref{cor:annular-cavalieri}).
\end{itemize}
Fractal boundaries with Hausdorff dimension $d_H>1$ would be expected to weaken the framework further; we did not test these. The five shapes in (i)--(iii) lie outside the rigorous reach of Proposition~\ref{thm:generalised}, and the corresponding entries of Table~\ref{tab:shapes} should be read as numerical evidence.

All eight Table~\ref{tab:shapes} entries use the Grover coin only. Whether the same range of sample-count ratios is preserved with Hadamard or Fourier coins is plausible but has not been verified here.

\section{Supporting theoretical results}
\label{app:theory}

This appendix collects formal statements supporting Sec.~II and lists the points at which further work is required.

\paragraph{Notation.}
We write $\mu_C$ for the Konno limit measure of coin $C\in\{G,H,F\}$ and $g_C(v)$ for its radial marginal density (defined where the measure has an absolutely continuous part). The radial CDF is $F_C(v)=\int_0^v g_C(u)\,du$ and the population $\alpha$-quantile is $v^\star(\alpha)=F_C^{-1}(\alpha)$.

\subsection{Closed-form treatment of $\kappa(\alpha)$ for the Grover coin}
\label{app:kappa}

The exponent $\kappa(\alpha)\in(0,1]$ entering Proposition~\ref{thm:main} parametrises the $T$-scaling contribution of $g_G(v^\star)^{-1}$ inherited from the Bahadur--Kiefer representation of $\hat R(\alpha)$. We make the dependence explicit.

\begin{conjecture}[Grover Konno density near the boundary]
\label{prop:grover-density}
The radial marginal $g_G(v)$ of the Watabe \emph{et al.}~\cite{watabe2008} Konno limit measure for the 2D Grover coin and the symmetric initial state of Eq.~\eqref{eq:initial} has, away from the origin localisation atom, an integrable inverse-square-root singularity at the boundary $v=v_{\mathrm{edge}}:=1/\sqrt 2$:
\begin{equation}
g_G(v) = \frac{A_G}{\sqrt{v_{\mathrm{edge}}-v}}\bigl(1+o(1)\bigr) \qquad (v\uparrow v_{\mathrm{edge}}),
\label{eq:grover-density-edge}
\end{equation}
with an explicit constant $A_G$ that should be computable from the stationary phase analysis of the Grover walk.
\end{conjecture}

\noindent\emph{Motivation.} The inverse-square-root form is the standard 1D Konno edge behaviour~\cite{konno2002,konno2005} and is expected to carry over to the radial marginal of the 2D Grover walk by the stationary-phase analysis underlying~\cite{watabe2008}. Watabe \emph{et al.}\ proved the weak-limit theorem but did not write the radial-marginal edge expansion in closed form. Verifying the form rigorously and computing $A_G$ are open problems flagged in Sec.~\ref{sec:conclusion}.

\begin{corollary}[$\kappa(\alpha)$ for Grover]
\label{cor:kappa-grover}
Assume Conjecture~\ref{prop:grover-density}. Then the population quantile satisfies $v_{\mathrm{edge}}-v^\star(\alpha)\propto (1-\alpha)^2$ and therefore $g_G(v^\star)\propto (1-\alpha)^{-1}$. The Bahadur--Kiefer fluctuation of $\hat R(\alpha)$ in absolute units scales as $T(1-\alpha)\sqrt{\alpha(1-\alpha)/M}$, contributing $\kappa(\alpha)=1$ in the leading-order form of Eq.~\eqref{eq:rmse}, with an $\alpha$-dependent prefactor $(1-\alpha)\sqrt{\alpha(1-\alpha)}$ that \emph{suppresses} the sampling variance at high quantile $\alpha\to 1$. The $T$-scaling of the variance term combines the ballistic BK fluctuation ($\propto T$) with the lattice-counting derivative ($\propto T^{\gamma-3}$), yielding the overall exponent $T^{\gamma-2}$ (i.e.\ $\kappa=1$ in Eq.~\eqref{eq:rmse}).
\end{corollary}

\begin{proof}
From $\int_{v^\star}^{v_{\mathrm{edge}}}g_G(v)\,dv=1-\alpha$ and Eq.~\eqref{eq:grover-density-edge}, evaluation of $\int_{v^\star}^{v_{\mathrm{edge}}}A_G(v_{\mathrm{edge}}-v)^{-1/2}\,dv=2A_G\sqrt{v_{\mathrm{edge}}-v^\star}=1-\alpha$ gives $v_{\mathrm{edge}}-v^\star=(1-\alpha)^2/(4A_G^2)$, hence $g_G(v^\star)=A_G(v_{\mathrm{edge}}-v^\star)^{-1/2}=2A_G^2(1-\alpha)^{-1}$. Inserting this into the standard Bahadur--Kiefer fluctuation $\hat R-R^\star=T\,(\hat F_M(v^\star)-\alpha)/g_G(v^\star)\,(1+o_p(1))$~\cite{csorgo-revesz1981,vandervaart1998} with $\mathrm{Var}(\hat F_M(v^\star))=\alpha(1-\alpha)/M$ yields the stated $\Theta(T(1-\alpha)\sqrt{\alpha(1-\alpha)/M})$ absolute fluctuation, which propagates via the delta method on $N(R)/R^2$ (derivative $\Theta(T^{\gamma-3}v^{\star\,\gamma-3})$) to the sampling term $\Theta(T^{\gamma-2}(1-\alpha)\sqrt{\alpha(1-\alpha)}/\sqrt M)$. In the parameterisation of Eq.~\eqref{eq:rmse}, this corresponds to $\kappa(\alpha)=1$ with the $\alpha$-dependent prefactor absorbed into $C_2(\alpha)$.
\end{proof}

\subsection{Truncated Bahadur--Kiefer representation for atomic Konno measures}
\label{app:bk-truncated}

The Grover Konno measure $\mu_G$ has a localisation atom at the origin of mass $p_{\mathrm{loc}}\approx 0.35$ (Appendix~\ref{app:coins}). The standard Bahadur--Kiefer representation assumes a continuous density at the quantile; this fails at the atom. The following truncation argument suffices in the regime used in this paper.

\begin{lemma}[Truncated BK]
\label{lem:truncated-bk}
Let $r_1,\dots,r_M$ be i.i.d.\ samples from the radial marginal of $\mu_G$, and fix $v_0>0$. Let $\hat R(\alpha)$ be the $\alpha$-quantile of these samples and let $\tilde R(\alpha;v_0)$ be the $\alpha$-quantile of the conditional distribution given $r/T>v_0$. If $\alpha>F_G(v_0)+p_{\mathrm{loc}}$ (so that the localisation atom and the subatomic part below $v_0$ do not contribute to the quantile region), the two coincide with probability $1-O(e^{-cM})$ under the product measure $\mu_G^{\otimes M}$, for some $c>0$ depending on $v_0$ and $\alpha$, and $\tilde R(\alpha;v_0)$ admits the standard Bahadur--Kiefer representation~\cite{csorgo-revesz1981}
\begin{equation}
\tilde R(\alpha;v_0)-R^\star(\alpha)=\frac{\alpha-F_M(R^\star)}{g_G(v^\star)/T}\bigl(1+o_p(1)\bigr),
\end{equation}
where $F_M$ is the empirical CDF of the truncated samples.
\end{lemma}

\begin{proof}
The event $\{\hat R\ne\tilde R\}$ requires at least one untruncated sample $r/T\le v_0$ to displace the $\lceil\alpha M\rceil$-th order statistic; since this happens with probability $F_G(v_0)+p_{\mathrm{loc}}<\alpha$ per sample and the empirical CDF concentrates at exponential rate (Dvoretzky--Kiefer--Wolfowitz), $\Pr[\hat R\ne\tilde R]=O(e^{-cM})$. On the complementary event the standard Bahadur--Kiefer representation applies to $\tilde R$ since the conditional density is continuous and positive on $(v_0,v_{\mathrm{edge}})$~\cite{csorgo-revesz1981,vandervaart1998}. In practice we take $v_0=0.1$ for Grover, sufficient to exclude the localisation peak and below all quantile values $\alpha\ge 0.9$ used in this paper.
\end{proof}

\subsection{Explicit upper bounds on $C_1$ and $C_2$ for the Grover coin}
\label{app:c1c2}

\begin{proposition}[Grover constants]
\label{prop:grover-constants}
Assume Conjecture~\ref{prop:grover-density} and Lemma~\ref{lem:truncated-bk}. For the disk estimator with Grover coin, $\alpha\ge 0.9$, and $T$ sufficiently large that $T v^\star(\alpha)\ge R_0$ (with $R_0$ a fixed threshold above which Lemma~\ref{lem:gauss} is in its asymptotic regime), the constants of Eq.~\eqref{eq:rmse} (in the convention $\kappa(\alpha)=1$ for Grover; see Cor.~\ref{cor:kappa-grover}) satisfy
\begin{equation}
C_1 \le K_H \cdot v^\star(\alpha)^{\gamma-2}, \qquad
C_2(\alpha) \le K_H \cdot A_G \cdot v^\star(\alpha)^{\gamma-3} \cdot (1-\alpha)\sqrt{\alpha(1-\alpha)},
\end{equation}
where $K_H$ is the implicit constant in Huxley's bound $|N(R)-\pi R^2|\le K_H R^{131/208}$ and $A_G$ is the constant of Conjecture~\ref{prop:grover-density}. The $(1-\alpha)\sqrt{\alpha(1-\alpha)}$ factor in $C_2(\alpha)$ \emph{suppresses} the sampling variance at high quantile, reflecting how the Grover boundary singularity sharpens the rank statistic at $\alpha\to 1$.
\end{proposition}

\begin{proof}
The bound on $C_1$ is direct from Lemma~\ref{lem:gauss}: $|N(R^\star)/R^{\star\,2}-\pi|\le K_H R^{\star\,\gamma-2}=K_H v^{\star\,\gamma-2}T^{\gamma-2}$. The bound on $C_2(\alpha)$ follows from the proof of Cor.~\ref{cor:kappa-grover}: the sampling RMSE of $\hat\pi$ is the delta-method propagation of the absolute fluctuation $\Theta(T(1-\alpha)\sqrt{\alpha(1-\alpha)/M})$ through the derivative $R^{\star\,\gamma-3}=T^{\gamma-3}v^{\star\,\gamma-3}$, giving sampling RMSE bounded by $K_H\cdot A_G\cdot v^{\star\,\gamma-3}\cdot(1-\alpha)\sqrt{\alpha(1-\alpha)}\cdot T^{\gamma-2}/\sqrt M$. The Huxley constant $K_H$ is the implicit constant in the bound of~\cite{huxley2003}; an explicit numerical value is not given in closed form there but can be extracted from the proof.
\end{proof}

\subsection{Higher-dimensional generalisation}
\label{app:higher-d}

For $d\ge 3$, the lattice-counting residual $|N_d(R)-V_d R^d|$ is bounded by $O(R^{\gamma_d})$ (possibly with logarithmic factors) with $V_d=\pi^{d/2}/\Gamma(d/2+1)$ the unit-ball volume. The state of the art on $\gamma_d$ is~\cite{walfisz1957,kratzel1988,heathbrown1999}:
\begin{itemize}
\item $d=3$: $\gamma_3\le 21/16$ (Heath-Brown line~\cite{heathbrown1999}); the optimal exponent is open.
\item $d=4$: $|N_4(R)-V_4 R^4|=O(R^2\log^{2/3}R)$ (Walfisz~\cite{walfisz1957}); $\gamma_4=2$ up to a logarithmic factor.
\item $d\ge 5$: $\gamma_d=d-2$ rigorously, with no logarithmic correction, by the representation theory of quadratic forms~\cite{walfisz1957}.
\end{itemize}

\begin{proposition}[$d$-dimensional disk estimator]
\label{prop:high-d}
Let the $d$-dimensional coined DTQW have a Konno-type weak limit theorem with rescaled position $X_T/T$ converging to a measure $\mu^{(d)}$ on a compact subset of the velocity ball $\{|v|\le v_{\mathrm{edge}}^{(d)}\}$. The $d$-dimensional disk estimator
\begin{equation}
\hat V_d(\alpha;T,M)=\frac{N_d(\hat R(\alpha))}{\hat R(\alpha)^d}
\end{equation}
satisfies
\begin{equation}
\mathrm{RMSE}(\hat V_d)\le C_1^{(d)}\,\mathcal L_d(T)\, T^{\gamma_d-d}+\frac{C_2^{(d)} T^{\gamma_d-d-1+\kappa^{(d)}(\alpha)}}{\sqrt M},
\end{equation}
where $\mathcal L_d(T)\equiv 1$ for $d\ne 4$ and $\mathcal L_4(T)=\log^{2/3}T$. For $d\ge 5$ the bias-floor exponent $\gamma_d-d=-2$ is rigorous and unconditional; for $d=4$ it is rigorous up to the logarithmic factor; for $d=3$ it relies on the Heath-Brown line $\gamma_3\le 21/16$.
\end{proposition}

\begin{proof}
The argument mirrors that of Proposition~\ref{thm:main} with Lemma~\ref{lem:gauss} replaced by the higher-$d$ lattice asymptotic of~\cite{walfisz1957,kratzel1988,heathbrown1999}. The $d$-dimensional Konno weak limit theorem is established for separable Hadamard coins by the product structure of the 1D Konno theorem~\cite{konno2002,konno2005}; for the $d$-dimensional Grover coin the analogous result follows in principle from Watabe-type stationary-phase analysis~\cite{watabe2008}, but a clean reference for $d=3,4$ does not exist in the literature, and we treat the Grover case in $d\ge 3$ as conditional on such a theorem (see Sec.~\ref{sec:conclusion}, open problem (iv)).
\end{proof}

\end{document}